\documentclass[11pt,draftcls,onecolumn]{IEEEtran}
\usepackage{graphicx}
\usepackage{cite}

\usepackage{chemarrow}
\usepackage{color}
\usepackage{algorithmic}
\usepackage{algorithm}
\usepackage{clrscode}
\usepackage{amsmath, amssymb, mathrsfs, amsfonts, amsthm}
\usepackage{epsfig, epstopdf}

\usepackage{caption}
\usepackage{subcaption}
\usepackage{stfloats}
\usepackage{enumerate}

\theoremstyle{theorem}
\newtheorem{theorem}{Theorem}
\newtheorem{lemma}{Lemma}

\theoremstyle{remark}

\newcommand{\E}{\mathbb{E}}

\date{}
\title{On the Capacity of Point-to-Point and Multiple-Access Molecular Communications with Ligand-Receptors\IEEEauthorrefmark{1} \footnote{\IEEEauthorrefmark{1}This paper has been presented in part at the IEEE International Symposium on Information Theory (ISIT), Wan Chai, Hong Kong, June 2015.}}

\author{\IEEEauthorblockN{Gholamali Aminian\IEEEauthorrefmark{2}, Maryam Farahnak Ghazani\IEEEauthorrefmark{2},
Mahtab Mirmohseni\IEEEauthorrefmark{2}, \\ Masoumeh Nasiri Kenari\IEEEauthorrefmark{2}, and Faramarz Fekri\IEEEauthorrefmark{3}}\\
\IEEEauthorblockA{\IEEEauthorrefmark{2}Sharif University of Technology, \IEEEauthorrefmark{3}Georgia Institute of Technology}}

\begin{document}
\maketitle
\begin{abstract}
In this paper, we consider the bacterial point-to-point and multiple-access molecular communications with ligand-receptors. For the point-to-point communication, we investigate common signaling methods, namely the Level Scenario (LS), which uses one type of a molecule with different concentration levels, and the Type Scenario (TS), which employs multiple types of molecules with a single concentration level. We investigate the trade-offs between the two scenarios from the capacity point of view. We derive an upper bound on the capacity using a Binomial Channel (BIC) model and the symmetrized Kullback-Leibler (KL) divergence. A lower bound is also derived when the environment noise is negligible. For the TS, we also consider the effect of blocking of a receptor by a different molecule type. Then, we consider multiple-access communications, for which we investigate three scenarios based on molecule and receptor types, i.e., same types of molecules with Different Labeling and Same types of Receptors (DLSR), Different types of Molecules and Receptors (DMDR), and Same types of Molecules and Receptors (SMSR). We investigate the trade-offs among the three scenarios from the total capacity point of view. We derive some inner bounds on the capacity region of these scenarios when the environment noise is negligible.
\end{abstract}

\section{Introduction}
Molecular communication (MC) has stimulated a great deal of interest because of its potential broad applications. There are different mechanisms for MC, among which diffusion is the most favorable, as it does not require any prior infrastructure. In diffusion-based systems, information might be encoded into the concentration, type, or releasing time of the molecules. For instance, in \cite{Airfler2011}, an on-off keying modulation is proposed where molecules are released only when the information bit is one. It is shown that if there is no interference from the previous transmission slots, the channel can be modeled by a Z-channel. In \cite{kuran2011 ,tepekule2014}, new modulation techniques based on multiple types of molecules are presented.
Two models for diffusion-based channels have been proposed, namely small and large scales. Diffusion process is viewed as a probabilistic Brownian motion in the small scale model, whereas it is described by deterministic differential equations in the large scale model. In this paper, we concentrate on the \emph{large scale} model which reflects the average effects of diffusion. However, to derive the large-scale diffusion capacity of MC, one has to deal with the reception process at the receiver side. Two reception models are considered for a passive receiver. The first model is a perfect absorber where the receiver absorbs the hitting molecule. The second model, which is more realistic, is the ligand-receptor binding receiver, where the hitting molecule is absorbed by the receptor with some binding probability, \cite{fekri2, atakan2014molecular}. The randomness in ligand-receptor binding process is modeled in \cite{pierobon2011noise} and a closed form solution for this modeling is derived by using Markov chains. Ligand-receptors are modeled by a Markov chain in \cite{fekri2}, by a discrete-time Markov model in \cite{CapacityN2}, and by a BIC for a bacterial colony in \cite{fekri1}. The BIC is defined by $P(y|x)={n \choose y} x^y (1-x)^{n-y}$ where the input is $x\in [0,1]$, the output is $y \in \lbrace 0,1,\ldots,n \rbrace$ and $n$, the number of trials, is a given natural number. Average and peak constraints on the input $x$ may exist. The capacity of this channel without average and peak constraints, for large values of $n$, behaves as $\frac{1}{2}\log{\frac{n}{2\pi e}}+\log{\pi}$ \cite{xie1997minimax}.
However, there is no explicit upper or lower bound on the capacity of the BIC when $n$ is not large enough. An algorithm for computing the capacity of the BIC was presented in \cite{komninakis2001capacity} using convex optimization methods.

On the other hand, the bacteria based multiple-access communications have been studied in \cite{atakan2008molecular,atakan2009single,liu2013molecular} for diffusion channel and ligand-receptor, where the transmitters use binary on-off keying modulations employing the same type of molecules but with different labeling. In these papers, the capacity of the multiple-access channel (MAC) is simply computed as the sum of the capacity of the channels between each transmitter and the receiver. In \cite{atakan2008molecular}, the expected concentration of bound molecules is computed. Then, by approximating the number of delivered molecules as the normally distributed random variable, the maximum detection probability is calculated, and based on the result, by modeling each user channel as a binary symmetric channel, the capacity of the channel is computed. In \cite{atakan2009single}, the channel randomness effect has been modeled by adding an additive Gaussian noise to the concentration of bound molecules. Then, by using the Gaussian channel model approximation, the capacity of each user channel is derived. In \cite{liu2013molecular}, the capacity of each user channel is computed by representing the diffusion channel as a binary test channel. In all these works, the average interference from the other transmitters is taken into account in calculating the binding probability.
In this paper, however in contrast to the previous works, we examine the instantaneous effect of the multiple-access interference instead of its average value. In the following, we first concentrate on a point-to-point molecular communication and evaluate its capacity and the upper and lower bounds. Then we consider three multiple-access scenarios and for each, we evaluate the capacity region and some inner bounds.

Our main contributions are as follows:
\begin{itemize}
\item \textbf{Point-to-Point Communication:} We investigate the trade-offs between two bacterial point-to-point communication scenarios for ligand-receptors with fixed total number of molecules and receptors: 
(a) multi-type molecular communication with a single concentration level, and (b) single-type molecular communication with multiple concentration levels. At the first glance, scenario (a) introduces new degrees of freedom and reduces the intersymbol interference (ISI). However, since the number of molecules per type (the power per type) reduces increasing the number of types, we should examine the benefit of using different types of molecules. To make the comparison between scenario (a) and (b), we adopt the model of \cite{fekri1} in this work. In addition, a Markov model for the interactions between different types of molecules near the receptor is presented and the capacity for this model is computed numerically.
\item \textbf{Upper and Lower Bounds for the BIC Capacity:} Using KL divergence bound of \cite{aminian}, we derive an upper bound on the capacity of the point-to-point BIC model under given average and peak constraints on the channel input (Theorem \ref{theorem1}). Based on numerical evidence, we believe that this upper bound works well in the low SNR regime (which can occur in MC systems). A lower bound is derived on the point-to-point BIC capacity under average and peak constraints in the case of no environment noise in Lemma~\ref{lemma1}.
\item \textbf{Multiple-Access Communication:} We investigate the trade-offs among three multiple-access bacterial communication scenarios for ligand-receptors with fixed total number of receptors: (a) Using the same molecule type with different labeling for different transmitters and one receptor type at the receiver (DLSR), (b) Using different molecule types for different transmitters and different receptor types at the receiver (DMDR), and (c) Using the same molecule type for the transmitters and one receptor type at the receiver (SMSR). Scenarios (a) and (c) share the receptors and introduce a new degree of freedom. However, the benefit of using different types of molecules in scenario (b) should be examined. Scenario (a) has also the advantage that the transmitters use a self-identifying label and therefore seems to have better performance than scenarios (b) and (c). To compare the three scenarios, we compute their total capacities numerically. By assuming two transmitters in Section \ref{innerboundsmac}, we derive some inner bounds on the capacity region of the three scenarios under average and peak constraints in the case of no environment noise.
\end{itemize}
All logarithms are in base $e$ in this paper. This paper is organized as follows: in Section \ref{sec:model}, we present the system model for point-to-point communication scenarios, whose capacities are discussed in Section \ref{Capacity analysis}. The interaction of molecules near the receptor is modeled in Subsection \ref{sec:block}. In Subsection \ref{sec:cap_upper}, a new upper bound on the capacity of the BIC is presented by considering peak and average constraints. Subsection \ref{sec:cap_lower} includes a lower bound on the capacity of the BIC by extending the Z-channel. In Section \ref{macnetwork}, three scenarios for multiple-access communication are presented, whose capacity regions and total capacities are discussed in Section \ref{capacityregionmac}. The achievable rates for these scenarios are provided in Subsection \ref{innerboundsmac}. Section \ref{Simulation} includes the numerical results, and finally concluding remarks are given in Section \ref{conclusion}.

\section{Point-to-Point System Model}\label{sec:model}
In this section, we describe two bacterial point-to-point communication scenarios with ligand-receptors.

\begin{figure}
\centering
\begin{subfigure}[b]{0.5\textwidth}
\centering
\includegraphics[scale=0.53]{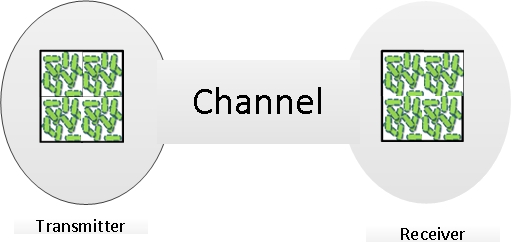}
\caption{Level scenario (LS)}
\label{fig1a}
\end{subfigure}%
\begin{subfigure}[b]{0.5\textwidth}
\centering
\includegraphics[scale=0.53]{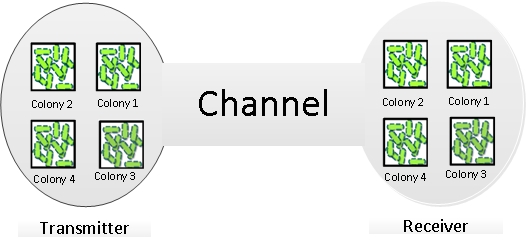}
\caption{Type scenario (TS) }
\label{fig1b}
\end{subfigure}
\caption{Two scenarios: LS and TS}
\label{fig1}
\vspace{-1em}
\end{figure}
\textbf{Level Scenario (LS):} Here, the transmitter encodes information at multiple concentration levels to create the codewords. At the transmitter and the receiver, there is only one colony with $n$ bacteria where each bacteria has $N$ receptors; i.e., $nN$ receptors in total. All these $n$ bacteria produce just one type of molecule. This scenario is shown in Fig.~\ref{fig1a}.

\textbf{Type Scenario (TS):} This scenario uses multiple types of molecules at the transmitter and the receiver. We assume the same total number of $n$ bacteria (as in LS) are available which are equally divided into $m$ colonies at both the transmitter and receiver as shown in Fig.~\ref{fig1b}. As such, each colony has $n/m$ bacteria. Moreover, different colonies at the transmitter produce different types of AHL molecules. Furthermore, the colonies are synchronized at the transmitter. Similar to the LS scenario, each bacteria has $N$ receptors. Therefore, there are ${nN}/{m}$ receptors in total per each colony, i.e., each type of molecule. Each colony can detect its own molecule type, and as a result, produces different color Fluorescent Proteins (e.g., GFP, YFP, ...) which are used by the receiver to decode the received signal. In addition, we assume that all receptors of a colony are independent and sense a common molecule concentration.

In both scenarios, we assume that there is no intersymbol interference (ISI). In other words, we assume those molecules, who do not bind to the receptors in the current time slot, will be degraded to the next time slot and hence will not interfere with molecules from the next time slot. This assumption, together with the large-scale diffusion channel property, results in a linear channel. For simplicity, we further assume that no attenuation occurs in the channel. Therefore, the received concentration $A_r$ is equal to the transmitted concentration $A_s$. At the receiver with ligand-receptors, the probability of binding at the steady state is given by \cite{fekri1}
\begin{align}\label{bind}
p_b=\frac{A_s}{A_s+\frac{\kappa}{\gamma}},
\end{align}
where $\gamma$ is the input gain and $\kappa$ is the dissociation rate of trapped molecules in the cell receptors. If we consider an environment noise with concentration $A_{ne}$, due to the molecules of the same type from other sources, the probability of binding becomes $p_b=\frac{A_s+A_{ne}}{A_s+A_{ne}+\frac{\kappa}{\gamma}}$.

In LS, we only have one type of molecule and its binding probability is equal to
\begin{align}\label{bindLS}
p_b^{LS}=\frac{X+A_{ne}^{LS}}{X+A_{ne}^{LS}+\frac{\kappa}{\gamma}},
\end{align}
where $X$ is the received concentration at the receiver and $A_{ne}^{LS}$ is the concentration of the environment noise.
We can view the LS scenario as a BIC as follows:
\begin{align}\label{BinomialchannelLS}
\centering
&P^{LS}(Y=y|X=x)= {nN \choose y} f_{p_b}^y(x+A_{ne}^{LS}) \left( 1-f_{p_b}(x+A_{ne}^{LS})\right)^{nN-y}, \\\nonumber
&f_{p_b}: [0,\infty] \rightarrow [0,1], \quad y \in \lbrace 0,1, ..., nN \rbrace.
\end{align}
The function $f_{p_b}(.)$ is the binding probability function. From \eqref{bindLS}, we have $f_{p_b}(X+A_{ne})=\frac{X+A_{ne}}{X+A_{ne}+\frac{\kappa}{\gamma}}$. As such, the function $f_{p_b}(.)$ is an increasing and concave function.

In TS, we have different types of molecules. Here, we assume that the binding processes of different molecule types are independent and every receptor binds to its own molecule type and two different types do not bind to one receptor. We investigate a more general model in Subsection~\ref{sec:block} by taking into account the interaction of different types of molecules in TS. The probability of binding for the $i$th type of molecule is given by
\begin{align}\label{bindTS}
p_{b_i}^{TS}=\frac{X_i+A_{ne_i}^{TS}}{X_i+A_{ne_i}^{TS}+\frac{\kappa_i}{\gamma_i}},
\end{align}
where $X_i$ is the received concentration of the $i$th type of molecule and $A_{ne_i}^{TS}$ is the concentration of the environment noise for the $i$th type of molecule. Without loss of generality, we assume $A_{ne_i}^{TS}=A_{ne}^{TS}$ and the same $\gamma$ and $\kappa$ for all types of molecules and receptors. This scenario can be viewed as $m$ orthogonal BICs as follows:
\begin{align}\label{BinomialchannelTS}
\centering
&P_i^{TS}(Y_i=y_i|X_i=x_i)={\frac{nN}{m} \choose y_i} f_{p_b}^{y_i}(x_i+ A_{ne_i}^{TS}) \left(1-f_{p_b}(x_i+A_{ne_i}^{TS})\right)^{\frac{nN}{m}-y_i}, \qquad i=1,...,m,\\\nonumber
&f_{p_b}: [0,\infty] \rightarrow [0,1], \quad y_i \in \lbrace 0,1, ..., \frac{nN}{m} \rbrace.
\end{align}


\subsection{Blocking of Receptors}\label{sec:block}
In the TS scenario, we assumed orthogonal parallel channels for different types of molecules with no interference between them (i.e., no blocking of a receptor by molecules of another type). However, when there are different types of molecules, they may interfere with each other. In other words, one type of molecule may block another type of molecule from binding to its receptor counterpart. For example, consider $m=2$ with two types of molecule, $A$ and $B$ and their corresponding receptors as $R_A$ and $R_B$. The molecule type $A$ near $R_B$ may prevent the molecule type $B$ from binding to $R_B$ and vice versa. Assume that $X_A$ and $X_B$ are the received concentrations of types $A$ and $B$. The main reaction kinetics, for binding of the molecule type $B$ to its receptor, is modeled as \cite{atakan2014molecular}
\begin{align}\label{kinetic}
X_B+R_B \underset{\kappa_{B}}{\overset {\gamma_{B}}{\rightleftharpoons}} XR_B,
\end{align}
where $\gamma_{B}\geq 0$ is the association rate of the molecule type $B$ with receptors of type $B$ and $\kappa_{B} \geq 0$ is the dissociation rate of $XR_B$ complex. Now, we characterize the blocking for the receptor type $B$, similar to the reaction kinetics formulas by
\begin{align}\label{blockingeq}
X_A+R_B \xrightarrow{\gamma_{B}^{Block,A}} R_B^{Block,A},\quad
X_A+R_B \xleftarrow{\kappa_{B}^{Block,A}} R_B^{Block,A},
\end{align}
where $\gamma_{B}^{Block,A} \geq 0$ is the blocking rate of $R_B$ by molecule type $A$ and $\kappa_{B}^{Block,A}$ is the unblocking rate of $R_B^{Block,A}$. If we do not take the blocking into account, then we have a reaction kinetics for each type of receptor to its molecule type. As in \cite{atakan2014molecular}, we define a Markov model for the no blocking case based on \eqref{kinetic}, as shown in Fig.~\ref{fig2a} for $m=2$. Likewise, according to \eqref{blockingeq}, we propose a Markov model for the blocking case, as shown in Fig.~\ref{fig2b}.
\begin{figure}
\quad
\begin{subfigure}[b]{0.42\textwidth}
\centering
\includegraphics[trim={0 0 0 0},scale=0.43]{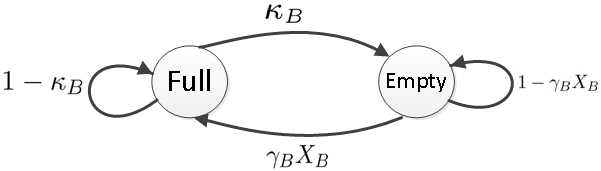}
\caption{With no blocking }
\label{fig2a}
\end{subfigure}\quad \qquad
\begin{subfigure}[b]{0.42\textwidth}
\centering
\includegraphics[trim={2.2cm 0 0 0},scale=0.4]{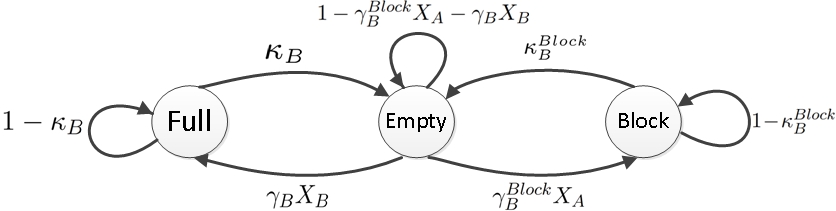}
\caption{With blocking}
\label{fig2b}
\end{subfigure}
\caption{Two Markov models for receptor type $B$}
\label{fig2}
\vspace{-1em}
\end{figure}
We consider three states. The full state is when the receptor binds to its type, the empty state when the receptor is free, and the block state when the receptor is blocked with a different molecule type.
Solving the chain for the no blocking case, the steady state behaviour of the system-reaction formula is obtained as \eqref{bind}. Solving the chain for the blocking case, we have the following probabilities of binding and blocking for the receptor type $B$:
\begin{align}\label{Blockequa}
&p_{b}=p_{Full}=\frac{\frac{\gamma_B}{\kappa_B} X_B}{\frac{\gamma_B}{\kappa_B} X_B + \frac{\gamma_B^{Block,A}}{\kappa_B^{Block,A}} X_A+1},\quad p_{Block}=\frac{\frac{\gamma_B^{Block,A}}{\kappa_B^{Block,A}} X_A}{\frac{\gamma_B}{\kappa_B} X_B + \frac{\gamma_B^{Block,A}}{\kappa_B^{Block,A}} X_A+1}.
\end{align}
If we increase the concentration for one type of molecule, the probability of binding for another type is decreased as expected. This model can be extended for $m>2$ via,
\begin{align}\label{blockbind}
&p_{b_i}=p_{Full_i}=\frac{\frac{\gamma_i}{\kappa_i} X_i}{\frac{\gamma_i}{\kappa_i} X_i + \sum_{j=1,j \neq i}^m \frac{\gamma_i^{Block,j}}{\kappa_i^{Block,j}}X_j + 1 },\quad p_{Block_i}=\frac{\sum_{j=1,j \neq i}^m \frac{\gamma_i^{Block,j}}{\kappa_i^{Block,j}}X_j}{\frac{\gamma_i}{\kappa_i} X_i + \sum_{j=1,j \neq i}^m \frac{\gamma_i^{Block,j}}{\kappa_i^{Block,j}}X_j + 1 }.
\end{align}
where $p_{b_i}$ and $p_{Block_i}$ are the binding probability of the $i$th type of receptor to the molecules of its type and the blocking probability of the $i$th type of receptor by the molecules of the other types, respectively. The blocking and unblocking rates for the $i$th type of receptor by the molecules of the $j$th type are defined by $\gamma_i^{Block,j}$ and $\kappa_i^{Block,j}$, respectively. It is also possible to consider the environment noise for the binding and blocking probabilities. Hence, the probability of binding for the $i$th type of molecule is given by
\begin{align}
&p_{b_i}^{TS,B}=\frac{\frac{\gamma_i}{\kappa_i} (X_i+A_{ne_i}^{TS})}{\frac{\gamma_i}{\kappa_i} (X_i+A_{ne_i}^{TS}) + \sum_{j=1,j \neq i}^m \frac{\gamma_i^{Block,j}}{\kappa_i^{Block,j}}(X_j+A_{ne_j}^{TS}) + 1 },
\end{align}
We can view the TS scenario with blocking as a multi-input multi-output BIC as follows:
\begin{align}\label{Mimobino}
\nonumber
&P_i^{TS,B}(Y_i=y_i|X_1=x_1, ..., X_m=x_m)\\
&\quad={\frac{nN}{m} \choose y_i} f_{p_{b_i}}^{y_i}(x_1, ..., x_m, A_{ne_i}^{TS}) \left(1-f_{p_{b_i}}(x_1, ..., x_m, A_{ne_i}^{TS})\right)^{\frac{nN}{m}-y_i}, \qquad i=1,...,m,
\end{align}
where $f_{p_{b_i}}(X_1,...,X_m,A_{ne_i}^{TS})=p_{b_i}^{TS,B}$ is the probability of binding when the blocking is considered.

\section{Point-to-point capacity analysis}
\label{Capacity analysis}
We investigate the capacity for the two scenarios. In both scenarios, the output is discrete. Further, we assume the environment noise and average and peak concentration level constraints.

In LS, we have a single colony with input $X$ and output $Y$. The peak and average concentration constraints for the input are $0 \leq X \leq A_s$ and $\E[X]\leq \alpha_s A_s$, respectively.

Then, we obtain the capacity for LS as
\begin{align}\label{LScap}
C_{LS}=\max_{\substack{P(x):\\ 0 \leq X\leq A_s,~ \E[X] \leq \alpha_s A_s}}I(X; Y), \qquad Y \in \lbrace 0,1, ..., nN \rbrace.
\end{align}

In TS, we use $X_i$ to denote the input of the $i$th colony to the channel and $Y_i$ to denote the output of the $i$th colony at the receiver. The peak and average concentration constraints for the input of the $i$th colony are $0 \leq X_i \leq \frac{A_s}{m}$ and $\E[X_i] \leq \alpha_s \frac{A_s}{m}$, respectively.

Hence, the capacity can be written as
\begin{align}\label{TScap}
&C_{TS}=\max_{\substack{P(x_1,x_2,...,x_m):\\ 0 \leq X_i\leq \frac{A_s}{m},~ \E[X_i] \leq \alpha_s \frac{A_s}{m}}} I(X_1,...,X_m;Y_1,...,Y_m), \qquad Y_i \in \lbrace 0,1, ..., \frac{nN}{m} \rbrace.
\end{align}
If we do not consider the blocking, the capacity could be obtaind as follows:
\begin{align}\label{TScap}
C_{TS} =m\times \max_{\substack{P(x_i):\\ 0 \leq X_i\leq \frac{A_s}{m},~ \E[X_i] \leq \alpha_s \frac{A_s}{m} } }I(X_i; Y_i),\qquad Y_i \in \lbrace 0,1, ..., \frac{nN}{m} \rbrace.
\end{align}

For a fair comparison of $C_{LS}$ with $C_{TS}$, we consider $A_{ne}^{LS}=A_{ne}^{TS}=A_{ne}$. Since we have a BIC in LS and $m$ BICs in TS with no blocking, we consider a BIC for the two scenarios as follows:
\begin{align}\label{binomchan}
&P(Y=y|X=x)= {N^\prime \choose y} f_{p_b}^y(x+A_{ne}) \left(1-f_{p_b}(x+A_{ne})\right)^{N^\prime-y},\\\nonumber
&f_{p_b}: [0,\infty] \rightarrow [0,1], \quad y \in \lbrace 0,1, ..., N^\prime \rbrace.
\end{align}
Since $P(y|x)$ is a Binomial distribution, we have $\sum_y yP(y|x)=N^\prime f_{p_b}(x)$.
The peak and average constraints for the input of the BIC are $0 \leq X \leq A_s^\prime$ and $\E[X]\leq \alpha_s A_s^\prime$, respectively. Note that for LS and TS we have the following parameters:
\begin{itemize}
\item \textbf{LS:} $N^\prime=nN$ and $A_s^\prime=A_s$.
\item \textbf{TS with no blocking:} $N^\prime=\frac{nN}{m}$ and $A_s^\prime=\frac{A_s}{m}$.
\end{itemize}
\subsection{Capacity Upper Bound}\label{sec:cap_upper}
There is no closed form for the BIC capacity. As such, for the first time, we propose an upper bound on the capacity of the BIC at the low SNR regime by considering average and peak constraints using the symmetrized KL divergence, referred as KL upper bound in \cite{aminian}. We first explain the KL upper bound briefly. Let $D_{\mathsf{sym}}(p\|q)=D(p\|q)+D(q\|p)$. Then,
\begin{align}\label{KL}
\mathcal{U}(P(y|x))&=\max_{P(x)}D_{\mathsf{sym}}(P(x,y)\|P(x)P(y))\geq \max_{P(x)}I(X;Y)=C(P(y|x)).
\end{align}
The KL $\mathcal{U}(P(y|x))$ is always an upper bound on the capacity. It is straightforward to show that
\begin{align}\label{KL1}
D_{\mathsf{sym}}\left(P(x,y)\|P(x)P(y)\right)=\E_{P(x,y)}\log P(Y|X) - \E_{P(x)P(y)}\log P(Y|X).
\end{align}
Now, we state our upper bound in the following theorem. The proof of this theorem can be found in Appendix \ref{AppendixProoftheorem1}.
\begin{theorem}\label{theorem1}
Consider a point-to-point BIC as \eqref{binomchan} and any input probability mass function (p.m.f) $P(x)$. Then, the symmetrized KL divergence upper bound has the following explicit formula:
\begin{align}
I(X;Y)&\leq \mathcal{U}(P(x,y))=N^\prime \mathsf{Cov} {\left(f_{p_b}(X+A_{ne}), \log \left( \frac{f_{p_b}(X+A_{ne})}{1-f_{p_b}(X+A_{ne})}\right)\right)},
\end{align}
where $\mathsf{Cov}(X,Y)=\E[XY]-\E[X]\E[Y]$. Furthermore, imposing the average intensity constraint $\alpha_s A_s^{\prime}$ and peak constraint $A_s^{\prime}$, we get
\begin{align}\label{upperbound}
\nonumber
\mathcal{U}_{\mathsf{Binomial}}(P(y|x))&:=\max_{\substack{P(x):\\ 0 \leq X \leq A_s^\prime,~ \E[X]=\alpha_s A_s^\prime}}\mathcal{U}(P(x,y))\\
&=N^\prime\begin{cases}\frac{f_{p_b}(\alpha_s A_s^\prime+A_{ne})}{f_{p_b}(A_s^\prime+A_{ne})}\left[f_{p_b}(A_s^\prime+A_{ne})-f_{p_b}(\alpha_s A_s^\prime+A_{ne})\right]E,&\textit{if}\quad (*),
\\
\frac{f_{p_b}(A_s^\prime+A_{ne})}{4} E,&\textit{if}\quad (**),
\end{cases}
\end{align}
where $E=\log{\left(\frac{f_{p_b}(A_s^\prime+A_{ne})(1-f_{p_b}(A_{ne}))}{f_{p_b}(A_{ne})(1-f_{p_b}(A_s^\prime+A_{ne}))}\right)}$, $(*):f_{p_b}(\alpha_s A_s^\prime+A_{ne})< \frac{f_{p_b}(A_s^\prime+A_{ne})}{2}$, and $(**):f_{p_b}(\alpha_s A_s^\prime+A_{ne})\geq \frac{f_{p_b}(A_s^\prime+A_{ne})}{2}$.
Hence,
\begin{align}
C=\max_{\substack{P(x):\\ 0 \leq X \leq A_s^\prime,~ \E[X]=\alpha_s A_s^\prime}}I(X;Y)\leq \mathcal{U}_{\mathsf{Binomial}}(P(y|x)).
\end{align}
\end{theorem}

We compute this KL upper bound numerically in Section~\ref{Simulation}. Based on the numerical evidence, this upper bound works well for all Binomial channels (such as MC channels) with low capacity.
\subsection{Capacity Lower Bound}\label{sec:cap_lower}
We obtain a lower bound on the capacity of the BIC when the environment noise is negligible. We assume a binary input, while in the previous section, a continuous input was assumed. Under this assumption, the resulted capacity is a lower bound on the capacity of the BIC. We compute a closed form formula for the lower bound in the following lemma.
\begin{lemma}\label{lemma1}
Consider a point-to-point BIC as \eqref{binomchan} and any input p.m.f $P(x)$, in which $A_{ne}=0$, $x \in \lbrace 0, A_s^\prime \rbrace$ and $\E[X] \leq \alpha_s A_s^\prime$. The capacity of this channel is obtained as
\begin{align}\label{Lowerbound}
C=\begin{cases} H\left(\frac{1}{1+e^{g(p_c)}}\right)-\frac{g(p_c)}{1+e^{g(p_c)}},\quad &\alpha_s \geq \frac{1}{1-p_c+e^{\frac{-p_c \log p_c}{1-p_c}}},
\\
f_I(\alpha_s,p_c), \quad &0<\alpha_s< \frac{1}{1-p_c+e^{\frac{-p_c \log p_c}{1-p_c}}}, 
\end{cases}
\end{align}
where $H(p)=-p\log p -(1-p) \log (1-p)$, $g(p)=\frac{H(p)}{1-p}$, $p_c=\left(\frac{\frac{\kappa}{\gamma}}{A_s^\prime+\frac{\kappa}{\gamma}}\right)^{N^\prime}$ and $f_I(\alpha,p)= -\alpha(1-p) \log{\alpha}+\alpha p \log{p}-(1-\alpha+\alpha p)\log{(1-\alpha+\alpha p)}.$
\end{lemma}
\begin{proof}
The proof is provided in Appendix \ref{AppendixProoflemma1}.
\end{proof}

If we consider $N^\prime=1$, then the channel would reduce to a Z-channel.
\section{Multiple-Access System Model}\label{macnetwork}
We describe three bacterial multiple-access communication scenarios with ligand-receptors based on molecule and receptor types differences. 

\begin{figure}
\begin{subfigure}{.33\linewidth}
\centering
\includegraphics[scale=0.4]{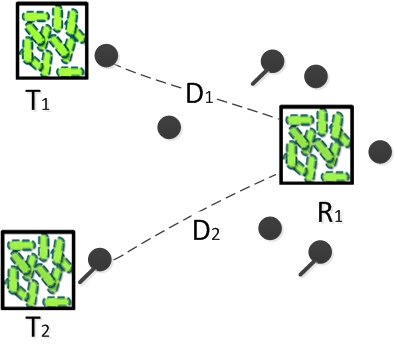}
\caption{DLSR scenario}
\label{fig3a}
\end{subfigure}%
\begin{subfigure}{.34\linewidth}
\centering
\includegraphics[scale=0.4]{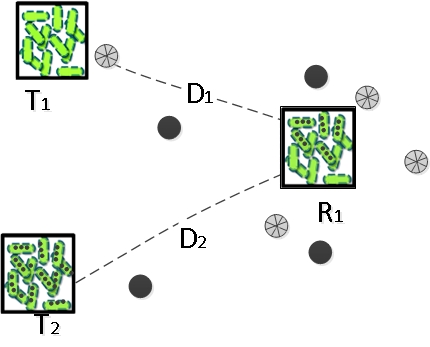}
\caption{DMDR scenario}
\label{fig3b}
\end{subfigure}%
\begin{subfigure}{0.33\linewidth}
\centering
\includegraphics[scale=0.4]{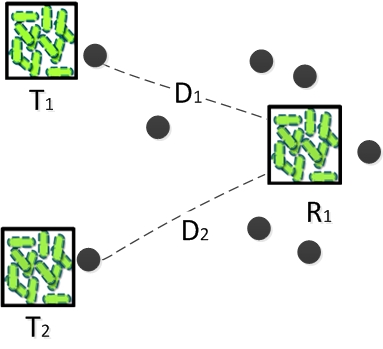}
\caption{SMSR scenario}
\label{fig3c}
\end{subfigure}
\caption{Three schemes of multiple-access in molecular communication systems}
\label{fig3}
\vspace{-1em}
\end{figure}
\textbf{DLSR Scenario:} As shown in Fig.~\ref{fig3a}, the transmitters send the same type of molecule (AHL) with different labelings and the receiver employs one type of bacteria (receptor). At the receiver, there is only one colony with $n$ bacteria where each bacteria has $N$ receptors; i.e., $nN$ receptors in total.

\textbf{DMDR Scenario:} As shown in Fig.~\ref{fig3b}, each transmitter uses a different type of bacteria and a different type of molecule (AHL) and the receiver employs different types of bacteria (receptor). At the receiver, there are $m$ different colonies with ${n}/{m}$ bacteria where each bacteria type has $N$ receptors; i.e., ${nN}/{m}$ receptors in total for the $i$th molecule type. 

\textbf{SMSR Scenario:} As shown in Fig.~\ref{fig3c}, the transmitters send the same type of molecule (AHL) and the receiver employs one type of bacteria (receptor). At the receiver, there is only one colony with $n$ bacteria, where each bacteria has $N$ receptors; i.e., $nN$ receptors in total.

In all scenarios, we assume that there is no intersymbol interference (ISI) and no attenuation occurs in the channel. Further, we assume that $X_i$ is the received concentration from the $i$th transmitter.

In the DLSR scenario, since different labelings are used\cite{atakan2008molecular,atakan2009single,liu2013molecular}, it is possible to distinguish between the molecules emitted from different transmitters. For example, consider $m=2$ with two different labelings of a molecule, $L_1$ and $L_2$. Assume that $X^{L_1}$ and $X^{L_2}$ are the received concentrations of the different labelings $L_1$ and $L_2$, respectively. The main reaction kinetics, for binding of the molecules with different labeling to the receptors, are modeled as
\begin{align}\label{eqkineticslabeling}
X^{L1}+R \underset{\kappa}{\overset {\gamma}{\rightleftharpoons}} XR^{L1},\quad
X^{L2}+R \underset{\kappa}{\overset {\gamma}{\rightleftharpoons}} XR^{L2},
\end{align}
where we consider the same association and dissociation rates for the two different labelings. Similar to the blocking case, we propose a Markov model for the labeling scenario, as shown in Fig.~\ref{fig4} for $m=2$.
\begin{figure}
\centering
\includegraphics[scale=0.45]{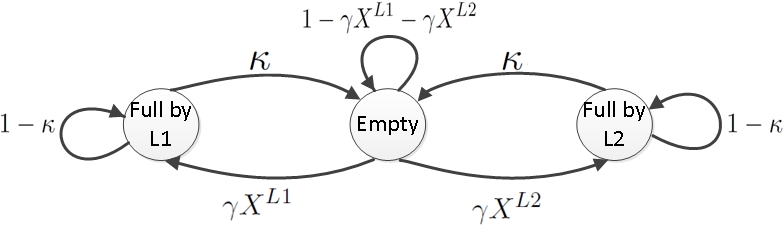}
\caption{Markov model for labeling}
\label{fig4}
\vspace{-1em}
\end{figure}
The steady state behaviour of the system-reaction formula is obtained as
\begin{align}
p_{b_1}=p_{Full~by~L_1}=\frac{X^{L_1}}{X^{L_1}+X^{L_2}+\frac{\kappa}{\gamma}},\quad
p_{b_2}=p_{Full~by~L_2}=\frac{X^{L_2}}{X^{L_1}+X^{L_2}+\frac{\kappa}{\gamma}}.
\end{align}
This model can be extended for $m>2$ via, 
\begin{align}
p_{b_i}=\frac{X^{L_i}}{\sum_{j=1}^{m}{X^{L_j}}+\frac{\kappa}{\gamma}},
\end{align}
where $p_{b_i}$ is the binding probability of the receptors to the molecules with the $i$th type of label.
It is also possible to consider the environment noise for the binding probabilities:
\begin{align}
p_{b_i}^{DLSR}=\frac{X_i+A_{ne_i}^{DLSR}}{{\sum_{j=1}^{m}(X_j+A_{ne_j}^{DLSR}})+\frac{\kappa}{\gamma}},
\end{align}
where $A_{ne_i}^{DLSR}$ is the concentration of the environment noise for the molecules with the $i$th type of label. Without loss of generality, we assume $A_{ne_i}^{DLSR}=A_{ne}^{DLSR}$. Let the output $Y_i$ be the number of receptors bound to the molecules with the $i$th type of label. The outputs have multinomial distribution with parameters $p_{b_1}^{DLSR}, ..., p_{b_m}^{DLSR}$:
\begin{align}
\nonumber
&P^{DLSR}{\left(y_1,...,y_m|x_1,...,x_m\right)}=P{\left(Y_1=y_1, ..., Y_m=y_m|X_1=x_1, ..., X_m=x_m\right)}\\
&\quad = \binom{nN}{y_1} ...\binom{nN-\sum_{i=1}^{m-1}{y_i}}{y_m}\left(p_{b_1}^{DLSR}\right)^{y_1} ...\left(p_{b_m}^{DLSR}\right)^{y_m} \left(1 -\sum_{i=1}^{m}{p_{b_i}^{DLSR}}\right)^{nN-\sum_{i=1}^{m}{y_i}}.
\end{align}

In the DMDR scenario, we have different molecule types for the transmitters. Without blocking, the binding probability for the $i$th type of molecule is obtained as
\begin{align}
p_{b_i}^{DMDR}=\frac{X_i+A_{ne_i}^{DMDR}}{X_i+A_{ne_i}^{DMDR}+\frac{\kappa_i}{\gamma_i}},
\end{align}
where $A_{ne_i}^{DMDR}$ is the concentration of the environment noise for the molecules of the $i$th type. Without loss of generality, we assume $A_{ne_i}^{DMDR}=A_{ne}^{DMDR}$. Let the output $Y_i$ be the number of receptors bound to the molecules of the $i$th type. Then, $Y_i \sim Binomial {\left(\frac{nN}{m}, p_{b_i}^{DMDR}\right)}$ and
\begin{align}
\nonumber
&P^{DMDR}{\left(y_1,...,y_m|x_1,...,x_m\right)}=P{\left(Y_1=y_1,...,Y_m=y_m|X_1=x_1,..., X_m=x_m\right)}\\
&\quad =\prod_{i=1}^{m}P{\left(Y_i=y_i|X_i=x_i\right)}=\prod_{i=1}^{m}{\binom{\frac{nN}{m}}{y_i} {\left(p_{b_i}^{DMDR}\right)}^{y_i}} {\left(1-p_{b_i}^{DMDR}\right)}^{\frac{nN}{m}-y_i}.
\end{align}
However, by considering the blocking, taking the same steps as deriving \eqref{blockbind}, we have the following binding probability for the $i$th type of molecule:
\begin{align}
p_{b_i}^{DMDR,B}=\frac{\frac{\gamma_i}{\kappa_i}(X_i+A_{ne_i}^{DMDR})}{\frac{\gamma_i}{\kappa_i}(X_i+A_{ne_i}^{DMDR})+\sum_{j=1,j \neq i}^{m}{\frac{\gamma_i^{Block,j}}{\kappa_i^{Block,j}}(X_j+A_{ne_j}^{DMDR})}+1},
\end{align}
Here, we have $Y_i \sim Binomial {\left(\frac{nN}{m}, p_{b_i}^{DMDR,B}\right)}$ and
\begin{align}
\nonumber
&P^{DMDR,B}{\left(y_1,...,y_m|x_1,...,x_m\right)}=P{\left(Y_1=y_1,...,Y_m=y_m|X_1=x_1,..., X_m=x_m\right)}\\
&\quad=\prod_{i=1}^{m}P{\left(Y_i=y_i|X_1=x_1,..., X_m=x_m\right)}
=\prod_{i=1}^{m}{\binom{\frac{nN}{m}}{y_i} {\left(p_{b_i}^{DMDR,B}\right)}^{y_i}} {\left(1-p_{b_i}^{DMDR,B}\right)}^{\frac{nN}{m}-y_i}.
\end{align}

In the SMSR scenario, we have one molecule type for the transmitters. The receiver senses the sum of the concentrations $X_i$. Hence, the probability of binding is equal to
\begin{align}
p_b^{SMSR}=\frac{\sum_{i=1}^{m}{X_i}+A_{ne}^{SMSR}}{\sum_{i=1}^{m}{X_i}+A_{ne}^{SMSR}+\frac{\kappa}{\gamma}},
\end{align}
where $A_{ne}^{SMSR}$ is the environment noise. Let the output $Y$ be the number of bound receptors. Then, $Y \sim Binomial{\left(nN , p_b^{SMSR}\right)}$ and
\begin{equation}
\begin{aligned}
P^{SMSR}{\left(y|x_1,...,x_m\right)}&=P{\left(Y=y|X_1=x_1,...,X_m=x_m\right)}=\binom{nN}{y} {\left(p_{b}^{SMSR}\right)}^{y} {\left(1-p_{b}^{SMSR}\right)}^{nN-y}.
\end{aligned}
\end{equation}

Table \ref{MACparameters} summarizes the variables defined in this section.
\begin{table} 
\centering
\caption{Variables of the multiple-access scenarios}
\begin{tabular}{l|l}
\hline \hline
Variable & Definition\\\hline 
$p_{b_i}^{DLSR}$ & $\frac{X_i+A_{ne_i}^{DLSR}}{{\sum_{j=1}^{m}(X_j+A_{ne_j}^{DLSR}})+\frac{\kappa}{\gamma}}$\\
$p_{b_i}^{DMDR}$ & $\frac{X_i+A_{ne_i}^{DMDR}}{X_i+A_{ne_i}^{DMDR}+\frac{\kappa_i}{\gamma_i}}$\\
$p_{b_i}^{DMDR,B}$ & $\frac{\frac{\gamma_i}{\kappa_i}(X_i+A_{ne_i}^{DMDR})}{\frac{\gamma_i}{\kappa_i}(X_i+A_{ne_i}^{DMDR})+\sum_{j=1,j \neq i}^{m}{\frac{\gamma_i^{Block,j}}{\kappa_i^{Block,j}}(X_j+A_{ne_j}^{DMDR})}+1}$\\ 
$p_b^{SMSR}$ & $\frac{\sum_{i=1}^{m}{X_i}+A_{ne}^{SMSR}}{\sum_{i=1}^{m}{X_i}+A_{ne}^{SMSR}+\frac{\kappa}{\gamma}}$\\
$P^{DLSR}(y_1,...,y_m|x_1,...,x_m)$ & $\binom{nN}{y_1} ...\binom{nN-\sum_{i=1}^{m-1}{y_i}}{y_m}\left(p_{b_1}^{DLSR}\right)^{y_1} ...\left(p_{b_m}^{DLSR}\right)^{y_m} \left(1-\sum_{i=1}^{m}{p_{b_i}^{DLSR}}\right)^{nN-\sum_{i=1}^{m}{y_i}}$\\ 
$P^{DMDR}(y_1,...,y_m|x_1,...,x_m)$ & $\prod_{i=1}^{m}{\binom{\frac{nN}{m}}{y_i} {\left(p_{b_i}^{DMDR}\right)}^{y_i}} {\left(1-p_{b_i}^{DMDR}\right)}^{\frac{nN}{m}-y_i}$\\
$P^{DMDR,B}(y_1,...,y_m|x_1,...,x_m)$ & $\prod_{i=1}^{m}{\binom{\frac{nN}{m}}{y_i} {\left(p_{b_i}^{DMDR,B}\right)}^{y_i}} {\left(1-p_{b_i}^{DMDR,B}\right)}^{\frac{nN}{m}-y_i}$\\
$P^{SMSR}(y|x_1,...,x_m)$ & $\binom{nN}{y} {\left(p_{b}^{SMSR}\right)}^{y} {\left(1-p_{b}^{SMSR}\right)}^{nN-y}$\\ 
\hline \hline
\end{tabular}
\label{MACparameters}
\end{table}

\section{Multiple-Access Capacity Region Analysis}\label{capacityregionmac}
In this section, we investigate the capacity region of the MAC for the three scenarios. In all scenarios, the output is discrete. Further, we assume the environment noise and consider peak and average concentration level constraints for the input of the $i$th transmitter as $0 \leq X_i \leq A_{s_i}$ and $\E[X_i]\leq \alpha_{s_i} A_{s_i}$.

The DMDR scenario with no blocking can be viewed as $m$ orthogonal point-to-point channels and the capacity of each channel can be computed according to Section \ref{Capacity analysis}. So here, we consider the blocking.

Since we have one receiver with $m$ outputs in the DLSR and DMDR scenarios, we may view these scenarios as SIMO (single transmit antenna and multiple receive antennas) MACs and compute the capacity region as the convex hull of rate tuples $(R_1,...,R_m)$ such that \cite{ElGamalKim}
\begin{align}\label{capacityregion_twoantennaMAC}
\sum_{i \in \mathcal{I}}{R_i} &\leq I(X(\mathcal{I});(Y_1,..,Y_m)|X(\mathcal{I}^c)) \qquad \forall\mathcal{I} \subseteq \{1,...,m \},
\end{align}
for some p.m.f $\prod_{i=1}^{k} P(x_i)$ that satisfies $0 \leq X_i\leq A_{s_i}$, $\E[X_i] \leq \alpha_{s_i} A_{s_i}$, $i=1,...,m$. The total capacity in these scenarios can be computed as follows:
\begin{align} \label{totalcapacity_MAC1}
C_{total}^{DLSR, DMDR}= \max_{\substack{P(x_1,x_2,...,x_m):\\ 0 \leq X_i\leq A_{s_i},~ \E[X_i] \leq \alpha_{s_i} A_{s_i},~i=1,...,m}}{I(X_1,...,X_m;Y_1,...,Y_m)}.
\end{align}

The SMSR scenario can be viewed as a SISO (single transmit antenna and single receive antenna) MAC. The capacity region of this channel is the convex hull of rate tuples $(R_1,...,R_m)$ such that \cite{ElGamalKim}
\begin{align} \label{capacityregion_MAC}
\sum_{i \in \mathcal{I}}{R_i} &\leq I(X(\mathcal{I});Y|X(\mathcal{I}^c)) \qquad \forall\mathcal{I} \subseteq \{1,...,m \},
\end{align}
for some p.m.f $\prod_{i=1}^{k} P(x_i)$ that satisfies $0 \leq X_i\leq A_{s_i}$, $\E[X_i] \leq \alpha_{s_i} A_{s_i}$, $i =1,...,m$. The total capacity in this scenario can be computed as follows:
\begin{align} \label{totalcapacity_MAC2}
C_{total}^{SMSR}=\max_{\substack{P(x_1,x_2,...,x_m):\\ 0 \leq X_i\leq A_{s_i},~ \E[X_i] \leq \alpha_{s_i} A_{s_i},~i= 1,...,m}}{I(X_1,...,X_m;Y)}.
\end{align}

There is no algorithm to compute the capacity region of the MAC numerically \cite{macreg}. Instead, the total capacities of the three scenarios are computed numerically in Sention \ref{Simulation}. We remark that the total capacity in the MAC is active and therefore it is sensible to compute it.

For a fair comparison of the total capacities, we consider $A_{ne}^{DLSR}=A_{ne}^{DMDR}=A_{ne}^{SMSR}=A_{ne}$.

\subsection{Capacity Region Inner Bounds}\label{innerboundsmac}
We consider two transmitters and obtain inner bounds on the capacity region of the multiple-access communication in the three scenarios when the environment noise is negligible. We assume a binary input to arrive at an inner bound, which is computed numerically in Section \ref{Simulation}.

\textbf{DLSR, DMDR:}
We may view the DLSR and DMDR scenarios as interference channels with full receiver cooperation. The capacity region of the interference channel is an inner bound on the capacity region of this channel. The time-division inner bound for an interference channel consists of all rate pairs $(R_1, R_2)$ such that
\begin{align}\label{timediv}
R_1 < k \ C_1,\quad
R_2 < (1-k) \ C_2,
\end{align}
for some $k \in [0,1]$, where $C_1$ and $C_2$ are the maximum achievable individual rates as follows \cite{ElGamalKim}:
\begin{align}
C_1=\max_{\substack{x_2,~ P(x_1)}} I(X_1;Y_1|X_2=x_2),\quad
C_2=\max_{\substack{x_1,~ P(x_2)}} I(X_2;Y_2|X_1=x_1).
\end{align}
This inner bound is computed in Lemma \ref{lemma2} for the DLSR and DMDR scenarios with binary inputs and considering peak and average concentration constraints. It is shown in this lemma, whose proof is provided in Appendix \ref{AppendixProoflemma2}, that the maximum achievable individual rate for each transmitter in the two scenarios occurs when the signal concentration of the other transmitter is zero and therefore the closed form formula for the maximum achievable individual rates is obtained.\\
The interference-as-noise inner bound for an interference channel consists of all rate pairs $(R_1, R_2)$ such that \cite{ElGamalKim}
\begin{align}
R_1 < I(X_1;Y_1),\quad 
R_2 < I(X_2;Y_2),
\end{align}
for some p.m.f $P(x_1)P(x_2)$. This inner bound is computed in Lemma \ref{lemma3} for the two scenarios with binary inputs and considering peak and average concentration constraints. The proof of this lemma is provided in Appendix \ref{AppendixProoflemma3}. 

\begin{lemma}\label{lemma2}
Consider interference channels with two sender-receiver pairs and $P^{DLSR}(y_1,y_2|x_1,x_2)$, $P^{DMDR,B}(y_1,y_2|x_1,x_2)$, and any input p.m.f $P(x_1)P(x_2)$, in which $A_{ne}^{DLSR}=A_{ne}^{DMDR}=0$, $x_1\in \{ 0, A_{s_1} \}$, $x_2 \in \{ 0, A_{s_2} \}$, $\E[X_1] \leq \alpha_{s_1} A_{s_1}$, and $\E[X_2] \leq \alpha_{s_2} A_{s_2}$. The time-division inner bound on the capacity region of these channels is obtained as
\begin{equation}
\begin{aligned}
&R_1 < k C_1, \quad R_2 < (1-k) C_2,\\
&C_i=\begin{cases} H\left(\frac{1}{1+e^{g(p_{c_{i0}})}}\right)-\frac{g(p_{c_{i0}})}{1+e^{g(p_{c_{i0}})}},\quad &\alpha_{s_i} \geq \frac{1}{1-p_{c_{i0}}+e^{\frac{-p_{c_{i0}} \log p_{c_{i0}}}{1-{p_c}_{i0}}}},\\
f_I(\alpha_{s_i}, p_{c_{i0}}), \quad &0 < \alpha_{s_i}< \frac{1}{1-p_{c_{i0}}+e^{\frac{-p_{c_{i0}} \log p_{c_{i0}}}{1-{p_c}_{i0}}}},
\end{cases}\qquad i =1,2,
\end{aligned}
\end{equation}
for some $k \in [0,1]$, where $H(p)=-p\log p-(1-p) \log{(1-p)}$, $g(p)=\frac{H(p)} {1-p}$, and $f_I(\alpha,p)= -\alpha(1-p) \log{\alpha}+\alpha p \log{p}-(1-\alpha+\alpha p)\log{(1-\alpha+\alpha p)}$. For the DLSR scenario, $p_{c_{i0}}=\left(\frac{\frac{\kappa}{\gamma}}{{A_s}_i+\frac{\kappa}{\gamma}}\right)^{nN}$, $i=1,2$ and for the DMDR scenario with blocking, $p_{c_{i0}}=\left(\frac{\frac{\kappa_i}{\gamma_i}}{ {A_s}_i+\frac{\kappa_i}{\gamma_i}}\right)^{\frac{nN}{2}}$, $i=1,2$.
\end{lemma}

\begin{lemma}\label{lemma3}
Consider interference channels with two sender-receiver pairs and $P^{DLSR}(y_1,y_2|x_1,x_2)$, $P^{DMDR,B}(y_1,y_2|x_1,x_2)$, and any input p.m.f $P(x_1)P(x_2)$, in which $A_{ne}^{DLSR}=A_{ne}^{DMDR}=0$, $x_1\in \{ 0, A_{s_1} \}$, $x_2 \in \{ 0, A_{s_2} \}$, $\E[X_1] \leq \alpha_{s_1} A_{s_1}$, and $\E[X_2] \leq \alpha_{s_2} A_{s_2}$. The interference-as-noise inner bound on the capacity region of these channels is obtained as
\begin{equation}
\begin{aligned}
R_i &<-\log{\alpha_i}+\alpha_i((1-\alpha_{j_i}) p_{c_{i0}}+\alpha_{j_i} p_{c_{i1}})\log{((1-\alpha_{j_i}) p_{c_{i0}}+\alpha_{j_i} p_{c_{i1}})}\\
&\quad -\alpha_i \left(\frac{1-\alpha_i}{\alpha_i}+(1-\alpha_{j_i}) p_{c_{i0}}+\alpha_{j_i} p_{c_{i1}}\right)\log{\left(\frac{1-\alpha_i}{\alpha_i}+(1-\alpha_{j_i}) p_{c_{i0}}+\alpha_{j_i} p_{c_{i1}}\right)}, \qquad i=1,2,
\end{aligned}
\end{equation}
for some $\alpha_1 \in [0, \alpha_{s_1}]$, $\alpha_2 \in [0, \alpha_{s_2}]$, where $j_1=2$ and $j_2=1$. For the DLSR scenario, $p_{c_{i0}}=\left(\frac{\frac{\kappa}{\gamma}}{A_{s_i}+\frac{\kappa}{\gamma}}\right)^{nN}$, $p_{c_{i1}}=\left(\frac{A_{s_{j_i}}+\frac{\kappa}{\gamma}}{A_{s_i}+A_{s_{j_i}}+\frac{\kappa}{\gamma}}\right)^{nN}$, $i=1,2$ and for the DMDR scenario with blocking, $p_{c_{i0}}=\left(\frac{\frac{\kappa_i}{\gamma_i}}{ A_{s_i}+\frac{\kappa_i}{\gamma_i}}\right)^{\frac{nN}{2}}$, $p_{c_{i1}}=\left(\frac{\frac{\gamma_i^{Block,j_i}}{\kappa_i^{Block,j_i}}A_{s_{j_i}}+1}{\frac{\gamma_i}{\kappa_i}A_{s_i}+\frac{\gamma_i^{Block,j_i}}{\kappa_i^{Block,j_i}}A_{s_{j_i}}+1}\right)^{\frac{nN}{2}}$, $i=1,2$, where $j_1=2$ and $j_2=1$.\\
For $A_{s_1}=A_{s_2}=A_s$, we have $p_{c_{10}}=p_{c_{20}}$ and $p_{c_{11}}=p_{c_{21}}$. Assume $\alpha_{s_1}=\alpha_{s_2}=\alpha_s$. The points where $R_1=R_2$ are obtained when $\alpha_1=\alpha_2$ and are computed as follows:
\begin{align}\label{maxeqachievablerates}
\nonumber
R_1=R_2=k&\bigg [-\log{\alpha^\prime}+\alpha^\prime((1-\alpha^\prime) p_{c_{10}}+\alpha^\prime p_{c_{11}}) \log{((1-\alpha^\prime) p_{c_{10}}+\alpha^\prime p_{c_{11}})}\\
&\quad -\alpha^\prime \left(\frac{1-\alpha^\prime}{\alpha^\prime}+(1-\alpha^\prime) p_{c_{10}}+\alpha^\prime p_{c_{11}}\right)\log{\left(\frac{1-\alpha^\prime}{\alpha^\prime}+(1-\alpha^\prime) p_{c_{10}}+\alpha^\prime p_{c_{11}}\right)}\bigg],
\end{align}
for some $k \in [0,1]$, where $\alpha^\prime=\min \{ \alpha, \alpha_s \}$ and $\alpha$ is the solution of the following equation:
\begin{align}
\nonumber
&((1-2 \alpha)p_{c_{10}}+2\alpha p_{c_{11}})\log{((1-\alpha) p_{c_{10}}+\alpha p_{c_{11}})}\\
&\quad-((1-2\alpha)p_{c_{10}}+2\alpha p_{c_{11}}-1)\log{\left(\frac{1-\alpha}{\alpha}+(1-\alpha) p_{c_{10}}+\alpha p_{c_{11}}\right)}=0.
\end{align}
\end{lemma}

\textbf{SMSR:} 
As mentioned before, we may view the SMSR scenario as a SISO MAC. According to \cite{ElGamalKim} for a MAC, the maximum achievable individual rates are
\begin{equation}
\begin{aligned}
C_1=\max_{\substack{x_2,~P(x_1)}} I(X_1;Y|X_2=x_2) ,\quad C_2=\max_{\substack{x_1,~ P(x_2)}} I(X_2;Y|X_1=x_1).
\end{aligned}
\end{equation}
Using these rates, the time-division inner bound can be obtained as \eqref{timediv}. This inner bound is computed in lemma \ref{lemma8} for the SMSR scenario with binary input and considering peak and average concentration constraints. The proof of this lemma is provided in Appendix \ref{AppendixProoflemma8}.
\begin{lemma}\label{lemma8}
Consider a MAC with two transmitters $P^{SMSR}(y|x_1,x_2)$ and any input p.m.f $P(x_1)P(x_2)$, in which $A_{ne}^{SMSR}=0$, $x_1\in \{ 0, A_{s_1} \}$, $x_2 \in \{ 0, A_{s_2} \}$, $\E[X_1] \leq \alpha_{s_1} A_{s_1}$, and $\E[X_2] \leq \alpha_{s_2} A_{s_2}$. The time-division inner bound on the capacity region of this channel is obtained as
\begin{equation}
\begin{aligned}
R_1&<k \max \{c_{10},c_{11}\}, \quad R_2<(1-k) \max \{c_{20},c_{21}\},\\
c_{i0}&=\begin{cases} H\left(\frac{1}{1+e^{g(p_{c_{i0}})}}\right)-\frac{g(p_{c_{i0}})}{1+e^{g(p_{c_{i0}})}},\quad &\alpha_{s_i} \geq \frac{1}{1-p_{c_{i0}}+e^{\frac{-p_{c_{i0}} \log p_{c_{i0}}}{1-{p_c}_{i0}}}}, \\
f_I(\alpha_{s_i},p_{c_{i0}}), \quad &0<\alpha_{s_i}< \frac{1}{1-p_{c_{i0}}+e^{\frac{-p_{c_{i0}} \log p_{c_{i0}}}{1-{p_c}_{i0}}}}, 
\end{cases} \qquad i=1,2,\\
c_{i1}&=-\sum_{l=0}^{nN} \left[(1-\alpha_i^\prime) P(y=l|x_i=0,x_{j_i}=A_{s_{j_i}})\log{\left((1-\alpha_i^\prime)+\alpha_i^\prime \frac{P(y=l|x_i=A_{s_i},x_{j_i}=A_{s_{j_i}})}{P(y=l|x_i=0,x_{j_i}=A_{s_{j_i}})}\right)}\right.\\
&\quad \left.+\alpha_i^\prime P(y=l|x_i=A_{s_i},x_{j_i}=A_{s_{j_i}}) \log{\left((1-\alpha_i^\prime) \frac{P(y=l|x_i=0,x_{j_i}=A_{s_{j_i}})}{P(y=l|x_i=A_{s_i},x_{j_i}=A_{s_{j_i}})}+\alpha_i^\prime \right)}\right],\qquad i=1,2,
\end{aligned}
\end{equation}
for some $k \in [0,1]$, where $j_1=2$, $j_2=1$, $H(p)=-p\log p-(1-p) \log{(1-p)}$, $g(p)=\frac{H(p)} {1-p}$, $f_I(\alpha,p)= -\alpha(1-p) \log{\alpha}+\alpha p \log{p}-(1-\alpha+\alpha p)\log{(1-\alpha+\alpha p)}$, $p_{c_{i0}}=\left(\frac{\frac{\kappa}{\gamma}}{A_{s_i}+\frac{\kappa}{\gamma}}\right)^{nN}$, $\alpha_i^\prime=\min\{\alpha_i, \alpha_{s_i}\}$,  $i=1,2$, where $\alpha_i$, $i=1,2$ is the solution of the following equation:
\begin{align}
\nonumber
\sum_{l=0}^{nN} &\left[P(y=l|x_i=0,x_{j_i}=A_{s_{j_i}})\log{\left((1-\alpha_i)+\alpha_i\frac{P(y=l|x_i=A_{s_i},x_{j_i}=A_{s_{j_i}})}{P(y=l|x_i=0,x_{j_i}=A_{s_{j_i}})}\right)}\right.\\
&\quad \left.-P(y=l|x_i=A_{s_i},x_{j_i}=A_{s_{j_i}}) \log{\left((1-\alpha_i) \frac{P(y=l|x_i=0,x_{j_i}=A_{s_{j_i}})}{P(y=l|x_i=A_{s_i},x_{j_i}=A_{s_{j_i}})}+\alpha_i\right)}\right]=0,
\end{align}
where $j_1=2$ and $j_2=1$.
\end{lemma}

\section{Numerical Results} \label{Simulation}
In this section, we first consider a point-to-point communication, and evaluate the rates for the TS and LS scenarios as well as the lower and KL upper bounds. Then, we evaluate the total capacity and achievable rates for the three scenarios of the multiple-access communications.
\subsection{Point-to-Point Capacity for LS and TS and Effect of Blocking}
We evaluate the rates of the TS scenario given in \eqref{TScap} and the LS scenario given in \eqref{LScap}, using the Blahut-Arimoto (BA) algorithm \cite{blahut}. The unit of the concentration of molecules is nano-Moles per litre (nM). We assume $N=10, n=16$, and use the values $\gamma=\gamma_1=...=\gamma_m=0.0004~\mbox{(nM~min)}^{-1}$ and $\kappa=\kappa_1=...=\kappa_m=0.1~{\mbox{min}}^{-1}$ from \cite{simval}. Note that we consider small values of $N$ and $n$ because of the time complexity of the BA algorithm for large values of $N$ and $n$, although in practice, these values can be very large.

Fig.~\ref{fig5a} shows the capacity of TS with no blocking and LS, for $m=2,4,8,16$ when $A_{ne}^{LS}=A_{ne}^{TS}=0$. It is seen that increasing the number of molecule types, $m$, from 1 improves the performance (for fixed $A_s$), which is expected due to the parallel transmission of the molecules. However, if we continue to increase $m$, and accordingly decrease the number of bacteria in each colony to ${n}/{m}$, the performance degrades. The reason is that decreasing the concentration level of TS in \eqref{bindTS} decreases the binding probability. Hence, there is an optimal $m$. For example, for $A_s=80$, this optimal value lies between $m=4$ and $m=8$. This implies that for $A_s=80$ and $m=2,4$, the capacity of TS is higher than LS, whereas for $m=8,16$, the capacity of TS is lower than LS. Similar conclusions can be made from Fig.~\ref{fig5b} in the presence of the environment noise $A_{ne}^{LS}=A_{ne}^{TS}=5$.
\begin{figure}
\centering
\begin{subfigure}[b]{0.5\textwidth}
\centering
\includegraphics*[trim={3cm 0 0 0}, scale=0.385]{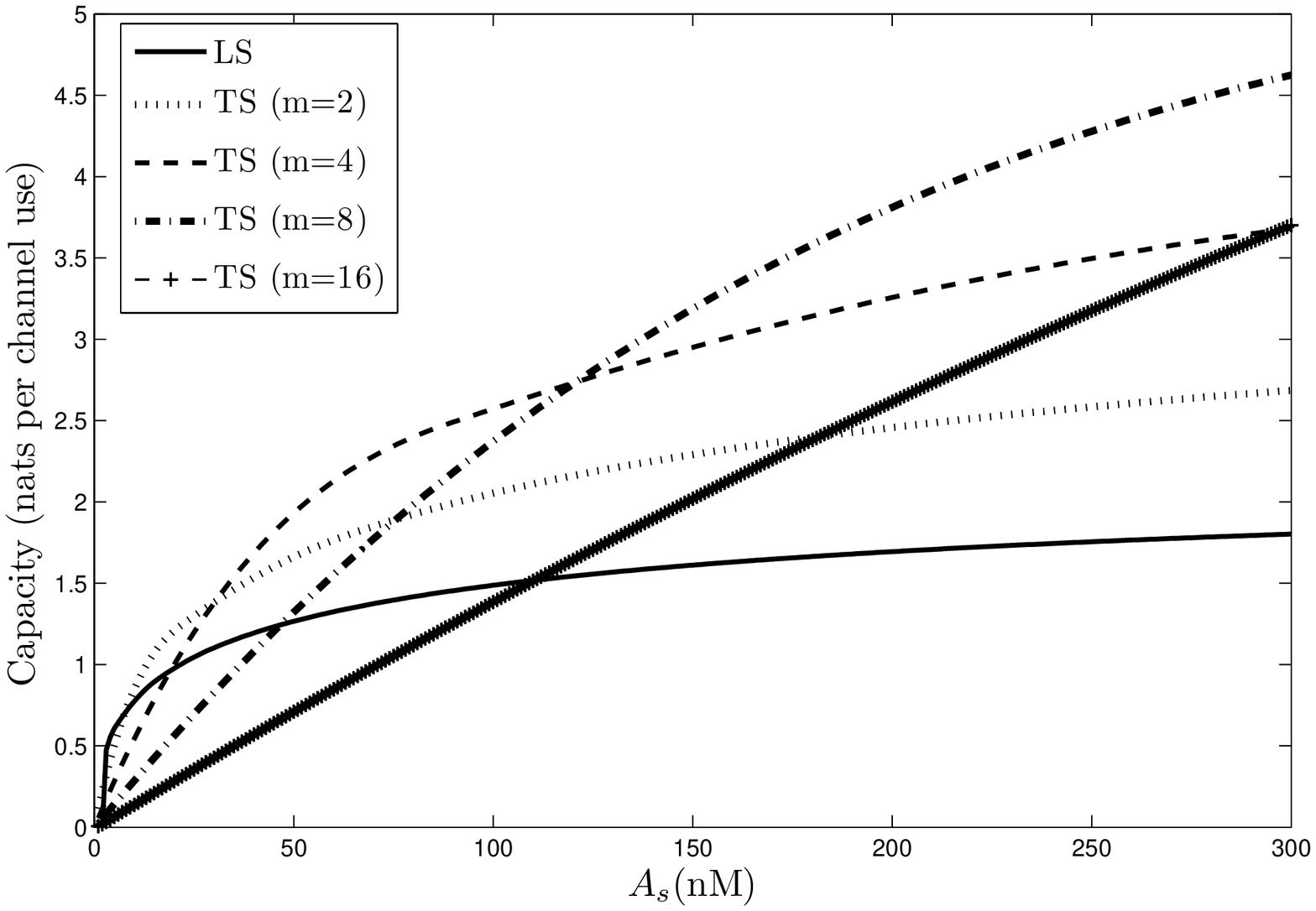}
\caption{$A_{ne}^{LS}=A_{ne}^{TS}=0$ }
\label{fig5a}
\end{subfigure}%
\begin{subfigure}[b]{0.5\textwidth}
\centering
\includegraphics*[trim={2.5cm 0 0 0},scale=0.38]{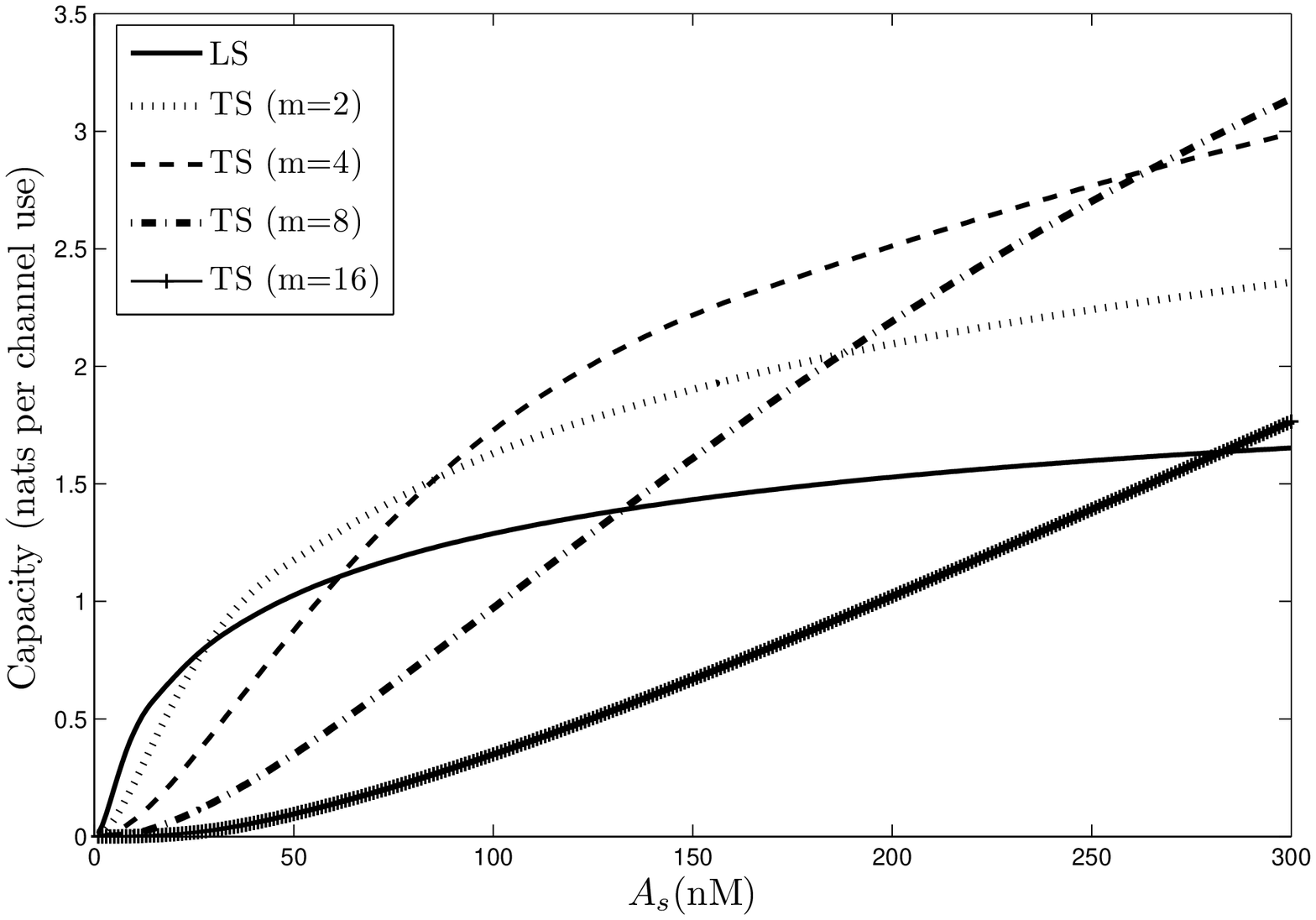}
\caption{$A_{ne}^{LS}=A_{ne}^{TS}=5$ }
\label{fig5b}
\end{subfigure}
\caption{Capacitiy of TS with no blocking and LS for $\alpha_s=\frac{1}{2}$.}
\label{fig5}
\vspace{-1em}
\end{figure}
\begin{figure}
\centering
\includegraphics*[trim={2cm 0 0 0},scale=0.6]{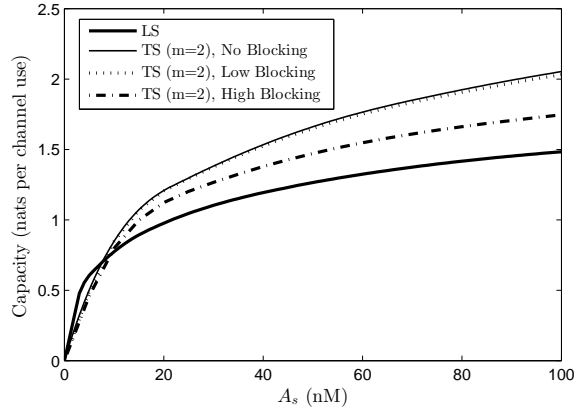}
\caption{Capacitiy of TS with and without blocking and LS for $\alpha_s=\frac{1}{2}$ and $A_{ne}=0$.}
\label{fig6}
\vspace{-1em}
\end{figure}

Fig.~\ref{fig6} shows the effect of blocking by showing the capacity of LS and TS for $m=2$. We considered two blocking cases:
\begin{itemize}
\item \textbf{Low Blocking:} $\gamma_1^{Block,2}=\gamma_2^{Block,1}=0.0003~\mbox{(nM~min)}^{-1}, \kappa_1^{Block,2}=\kappa_2^{Block,1}=0.15~\mbox{min}^{-1}.$
\item \textbf{High Blocking:} $\gamma_1^{Block,2}=\gamma_2^{Block,1}=0.0005~\mbox{(nM~min)}^{-1}, \kappa_1^{Block,2}=\kappa_2^{Block,1}=0.01~\mbox{min}^{-1}.$
\end{itemize}
As illustrated, the blocking decreases the capacity of TS. For small values of $A_s$, LS outperforms TS in all cases of blocking.

\subsection{Lower Bound and KL Upper Bound on the Capacity of the Point-to-Point Channel}
Our proposed KL upper bound, \eqref{upperbound}, and the capacity are depicted in Fig.~\ref{fig7} by considering the logarithmic scale. It can be observed that the distance between the KL upper bound and the capacity is constant in the logarithmic scale. Therefore, the gap between the capacity and the upper bound decreases as the environment noise increases. The lower bound in \eqref{Lowerbound} along with the capacity are shown in Fig.~\ref{fig8}. For simplicity, we consider average constraint to be inactive. For small values of $A_s^\prime$, our lower bound is tight which means the binary distribution is a capacity achieving distribution for small values of $A_s^\prime$.

\begin{figure}
\centering
\begin{minipage}{0.47\textwidth}
\centering
\includegraphics[trim={3cm 0 0 0},scale=0.35]{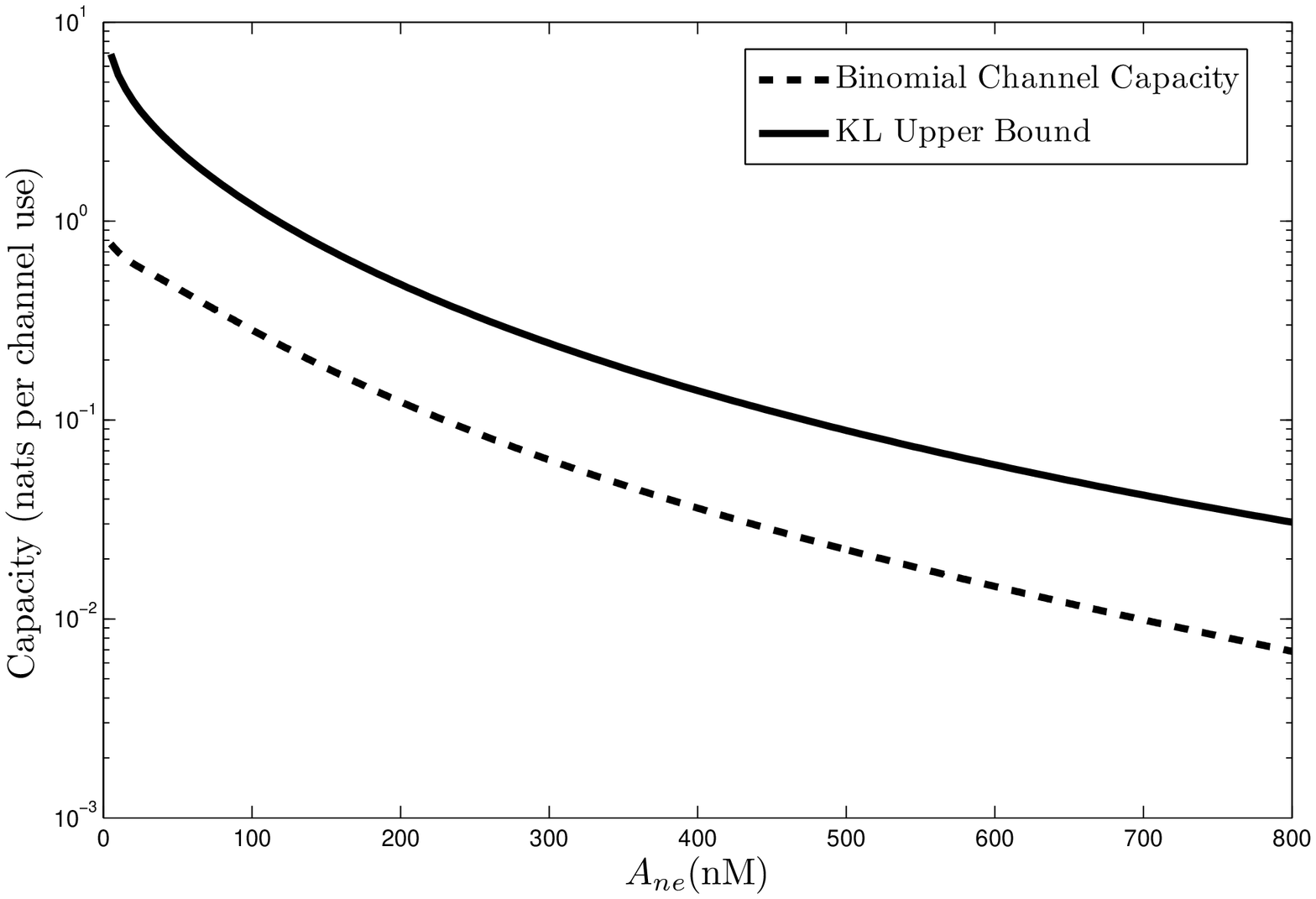}
\caption{Capacity and KL upper bound in terms of $A_{ne}$ for the BIC with $A_s^\prime=80$ and $\alpha_s=\frac{1}{2}$.}
\label{fig7}
\vspace{-1em}
\end{minipage}\qquad%
\begin{minipage}{0.47\textwidth}
\centering
\includegraphics*[trim={3.3cm 0 0 0},scale=0.445]{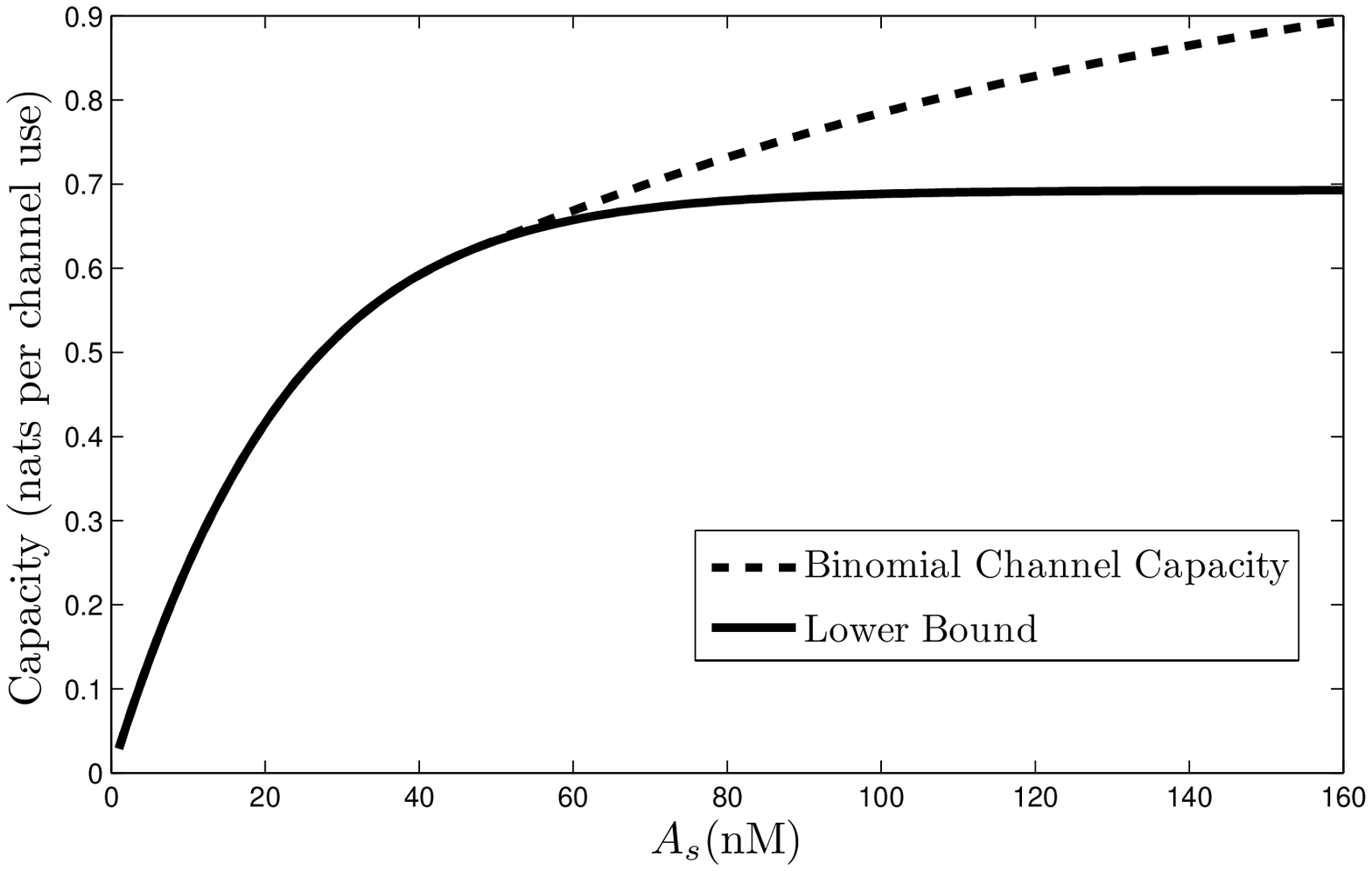}
\caption{Capacity and Lower Bound in terms of $A_s^\prime$ for the BIC with $N^\prime=20$ and $\alpha_s=\frac{1}{2}$.}
\label{fig8}
\vspace{-1em}
\end{minipage}
\end{figure}

\subsection{Multiple-Access Total Capacity}
In this section, we evaluate the total capacities of the DLSR and DMDR scenarios given in \eqref{totalcapacity_MAC1} and the SMSR scenario given in \eqref{totalcapacity_MAC2}, using the extension of the BA algorithm for the total capacity of the MAC \cite{extentionblahut}. We assume $N=10$, $n=6$, $m=2$. Similar to the previous sections, we use the values $\gamma=\gamma_1=\gamma_2=0.0004~\mbox{(nM~min)}^{-1}$, $\kappa=\kappa_1=\kappa_2=0.1~\mbox{min}^{-1},$ and consider no, low, and high blocking cases.

Fig. \ref{fig61An0-gs} shows the total capacities of the three scenarios in terms of $A_{s_1}=A_{s_2}=A_s$ when $A_{ne}^{DLSR}=A_{ne}^{DMDR}=A_{ne}^{SMSR}=0$. It is observed that DLSR has the highest total capacity for all values of $A_s$. For small values of $A_s$, SMSR has higher total capacity than DMDR, whereas for large values of $A_s$, SMSR has lower total capacity than DMDR. The reason is that when $A_s$ is small, sharing the receptors is useful. But when $A_s$ increases, using different types of molecules becomes more useful. Since DLSR has both of these advantages, it is more effective than the other two scenarios. Fig. \ref{fig61An5-gs} shows the total capacity for the three scenarios when $A_{ne}^{DLSR}=A_{ne}^{DMDR}=A_{ne}^{SMSR}=5$. Similar conclusions can be made in the presence of the environment noise.

\begin{figure}
\centering
\begin{subfigure}[b]{0.5\textwidth}
\centering
\includegraphics*[trim={2.9cm 0 0 0},scale=0.6]{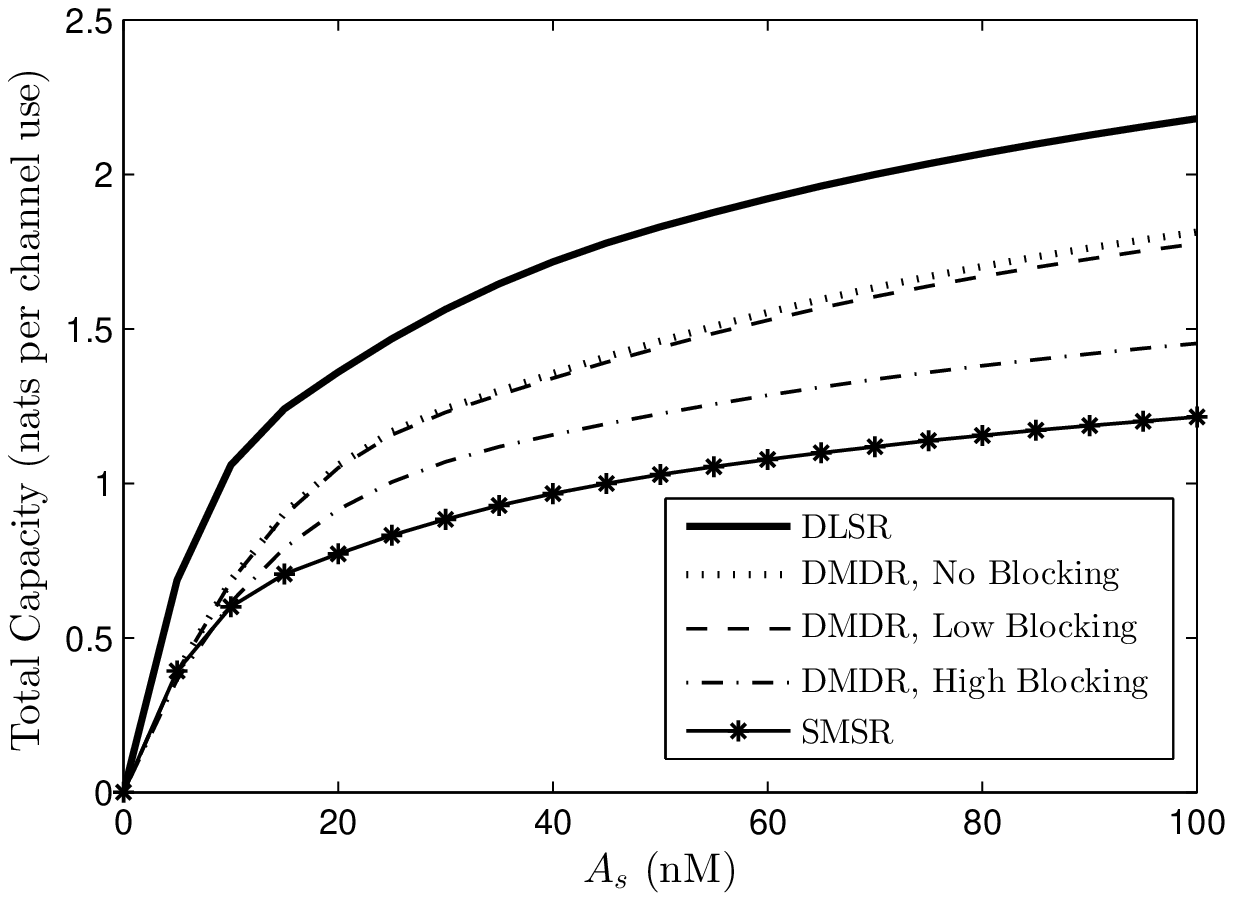}
\caption{$A_{ne}^{DLSR}=A_{ne}^{DMDR}=A_{ne}^{SMSR}=0$}
\label{fig61An0-gs}
\end{subfigure}%
\begin{subfigure}[b]{0.5\textwidth}
\centering
\includegraphics*[trim={2.9cm 0 0 0},scale=0.6]{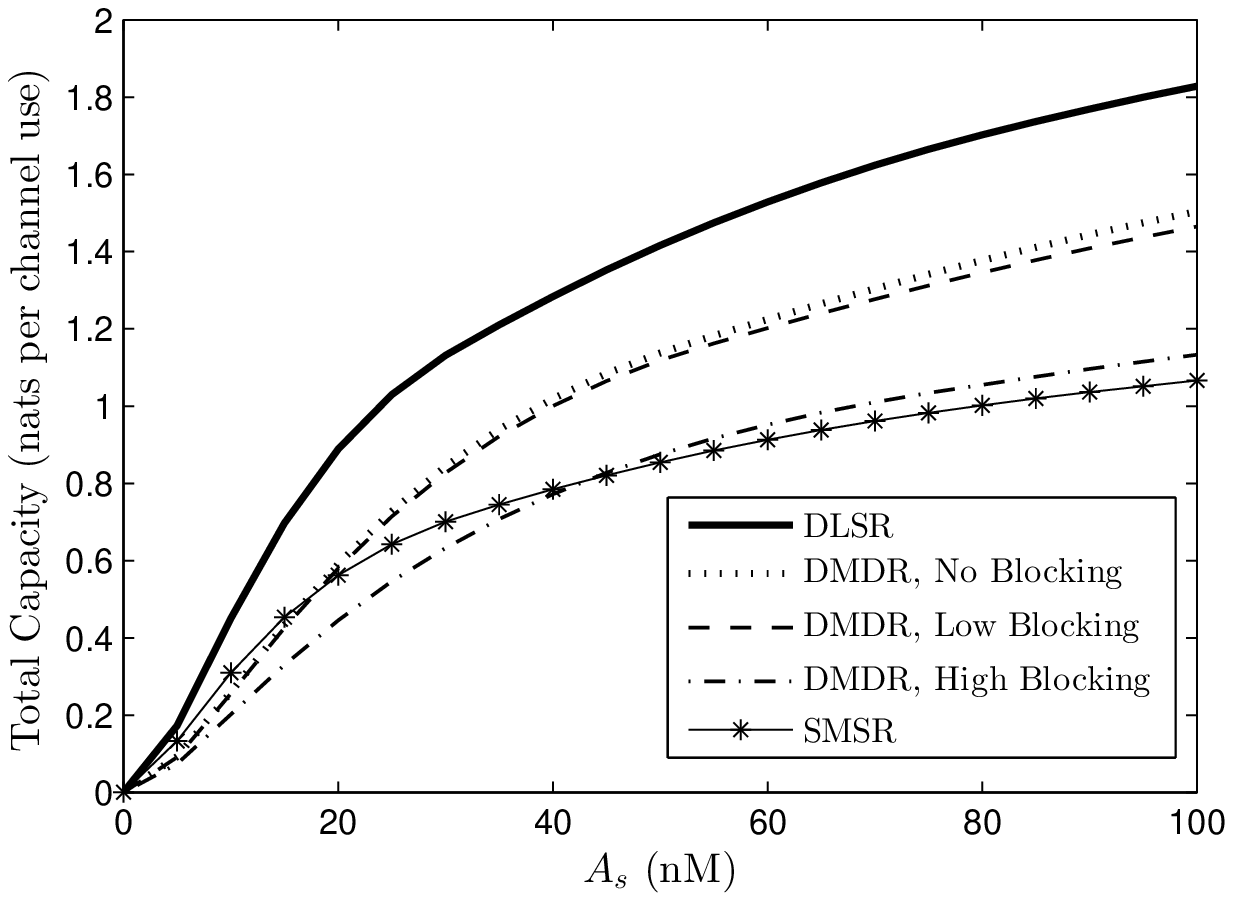}
\caption{$A_{ne}^{DLSR}=A_{ne}^{DMDR}=A_{ne}^{SMSR}=5$}
\label{fig61An5-gs}
\end{subfigure}
\caption{Total capacity of DLSR, DMDR, and SMSR for $\alpha_{s_1}=\alpha_{s_2}=\frac{1}{2}$.}
\label{figtotalcapacity}
\vspace{-1em}
\end{figure}
\begin{figure}
\centering
\includegraphics[trim={2cm 0 0 0},scale=0.5]{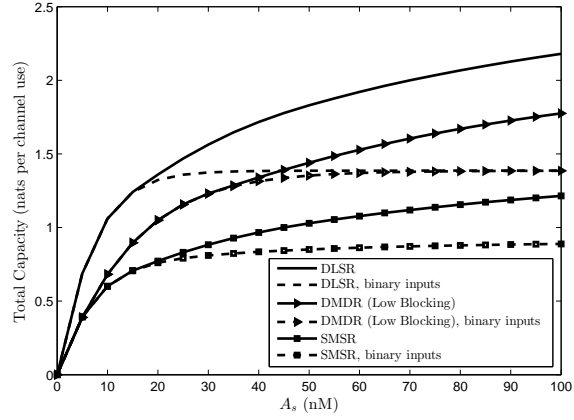}
\caption{Total capacity of DLSR, DMDR, and SMSR with continuous and binary inputs for $\alpha_{s_1}=\alpha_{s_2}=\frac{1}{2}$ and $A_{ne}=0$.}
\label{fig9gs}
\vspace{-1em}
\end{figure}
The total capacities of the three scenarios for both continuous and binary inputs are depicted in Fig.~\ref{fig9gs}. It can be observed that in all three scenarios, the total capacities with binary inputs are equal to the total capacities with continuous inputs for small values of $A_s$. For large values of $A_s$, the total capacities of DLSR and DMDR with binary inputs reach to the same value since all receptors become full and these scenarios behave the same. However, the total capacity of SMSR with binary input reaches to a lower value since it doesn't have the advantage of using different types of molecules or self-identifying labels.

\subsection{Inner Bounds on the Capacity Region of the MAC}
The capacity region inner bounds for the DLSR, DMDR, and SMSR scenarios, provided in Section \ref{innerboundsmac}, are depicted in Fig.~\ref{fig2-gs}. The capacity regions of the three scenarios with binary inputs are shown in Fig. \ref{fig5-gs} by considering $A_{s_1}=A_{s_2}=100$.
\begin{figure}
\centering
\begin{subfigure}[b]{0.5\textwidth}
\centering
\includegraphics*[trim={1.8cm 0 0 0},scale=0.61]{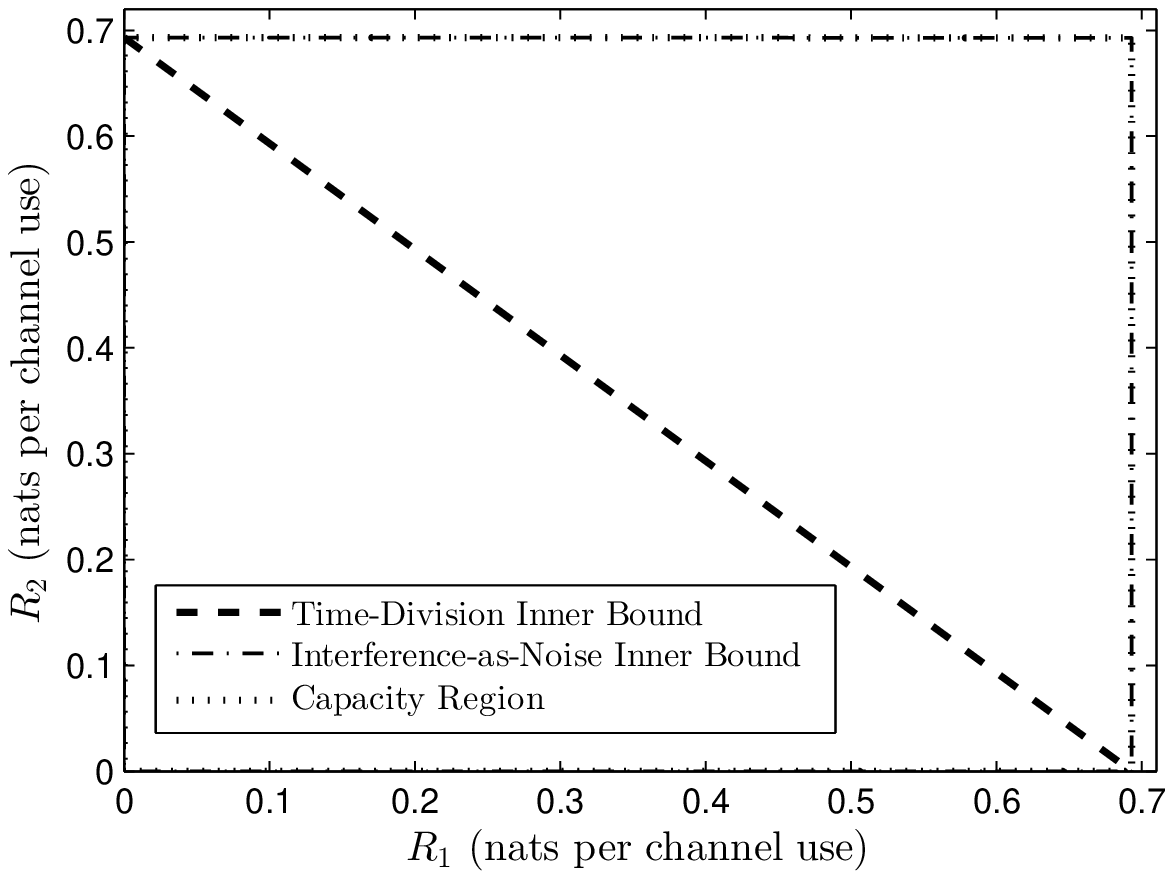}
\caption{DLSR}
\label{fig21-gs-1}
\end{subfigure}%
\begin{subfigure}[b]{0.5\textwidth}
\centering
\includegraphics*[trim={1.8cm 0 0 0},scale=0.61]{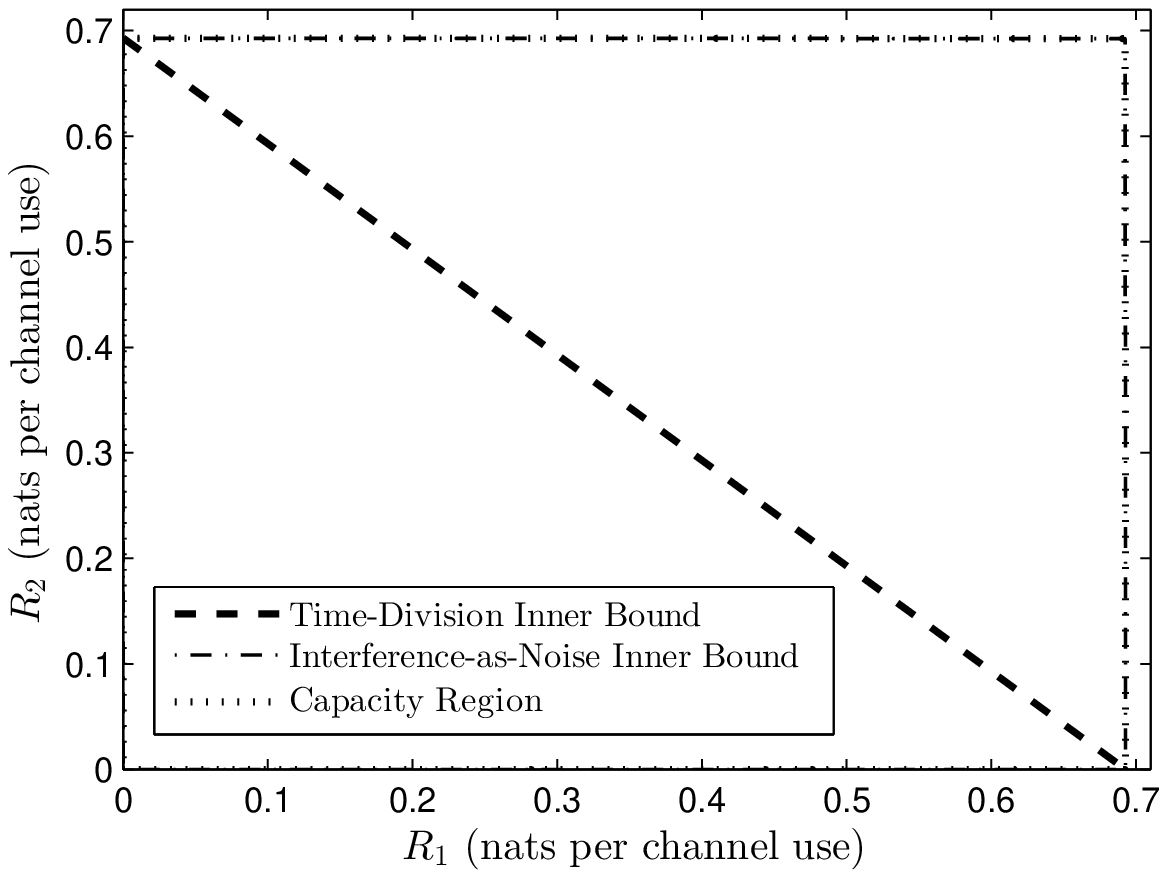}
\caption{DMDR with low blocking}
\label{fig21-gs-2}
\end{subfigure}\\
\begin{subfigure}[b]{0.5\textwidth}
\centering
\includegraphics*[trim={1.8cm 0 0 0},scale=0.61]{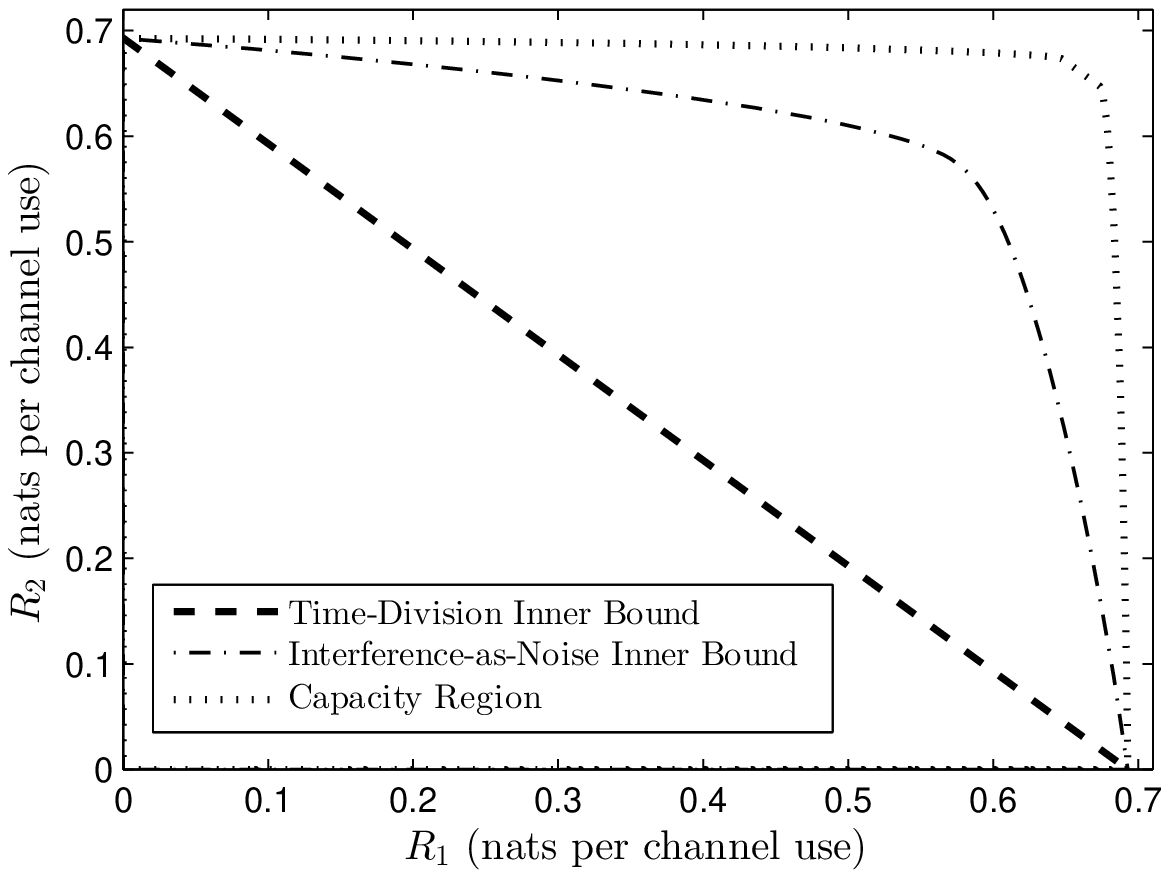}
\caption{DMDR with high blocking}
\label{fig22-gs}
\end{subfigure}%
\begin{subfigure}[b]{0.5\textwidth}
\centering
\includegraphics*[trim={1.8cm 0 0 0},scale=0.61]{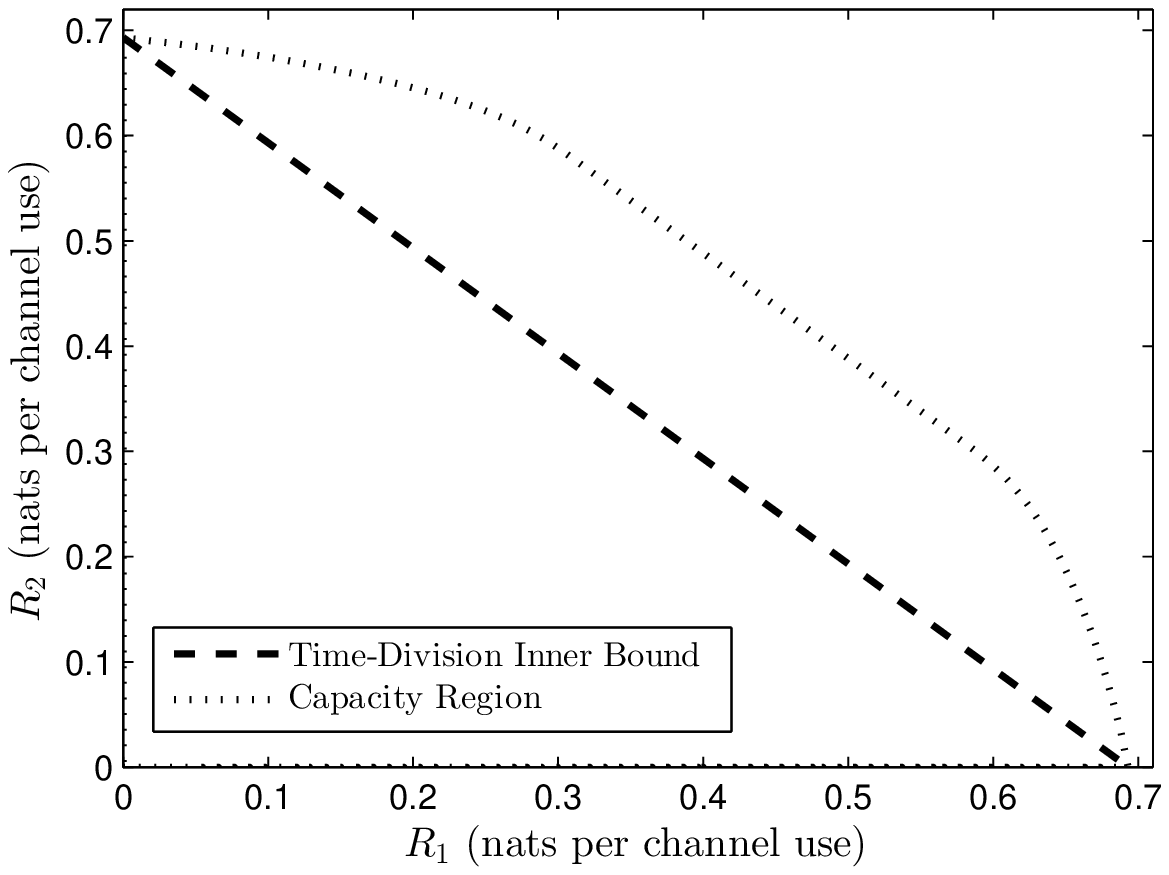}
\caption{SMSR}
\label{fig21-gs-3}
\end{subfigure}
\caption{Capacity region and Inner bounds for DLSR, DMDR, and SMSR with binary inputs for $\alpha_{s_1}=\alpha_{s_2}=\frac{1}{2}$ and $A_{ne}=0$.}
\label{fig2-gs}
\vspace{-1em}
\end{figure}
It is observed that DMDR with low blocking and DLSR have the same square shaped capacity regions, which indicates that for this parameter setup, these scenarios almost experience orthogonal MACs. These two scenarios have the largest Capacity region and SMSR has the smallest capacity region and the capacity region of DMDR with high blocking is in between.

Fig.~\ref{fig42-gs} shows the maximum achievable equal rates given in \eqref{maxeqachievablerates}, when considering interference as noise, in terms of $A_s$ for DMDR with low and high blocking and DLSR. It is observed that the rate for DLSR is larger than DMDR and reaches to a constant value as $A_s$ increases. Though the constant value is almost the same for DMDR with low blocking and DLSR, the value is higher than that of DMDR with high blocking. The reason is that when considering binary inputs and increasing $A_s$, DMDR with low blocking behaves like DLSR since all receptors become full. However, DMDR with high blocking behaves worse than DLSR since some of the receptors are always blocked.

\begin{figure}
\centering
\begin{minipage}{0.45\textwidth}
\centering
\includegraphics[trim={2.4cm 0 0 1.15cm},scale=0.615]{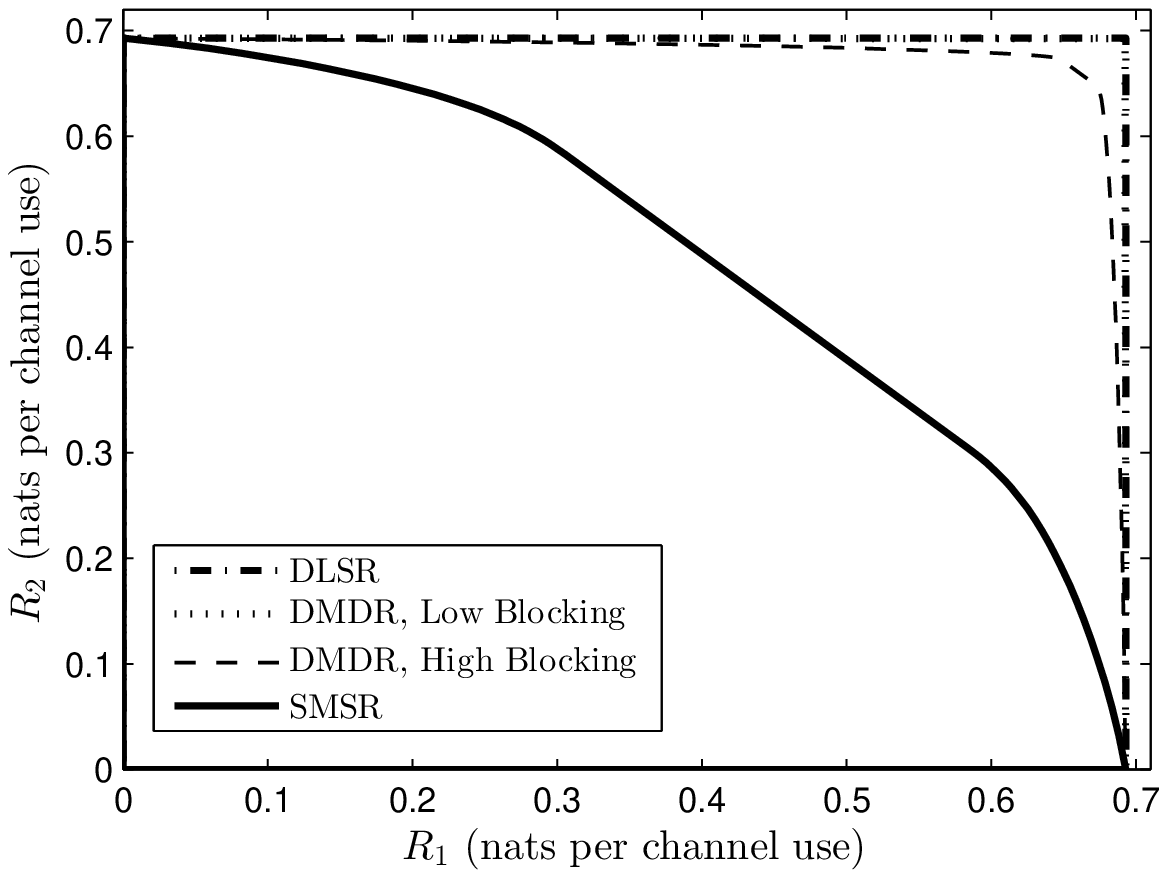}
\caption{Capacity region of DLSR, DMDR, and SMSR with binary inputs for $\alpha_{s_1}=\alpha_{s_2}=\frac{1}{2}$ and $A_{ne}=0$.}
\label{fig5-gs}
\vspace{-1em}
\end{minipage}\qquad %
\begin{minipage}{0.45\textwidth}
\centering
\includegraphics[trim={2.1cm 0 0 0},scale=0.61]{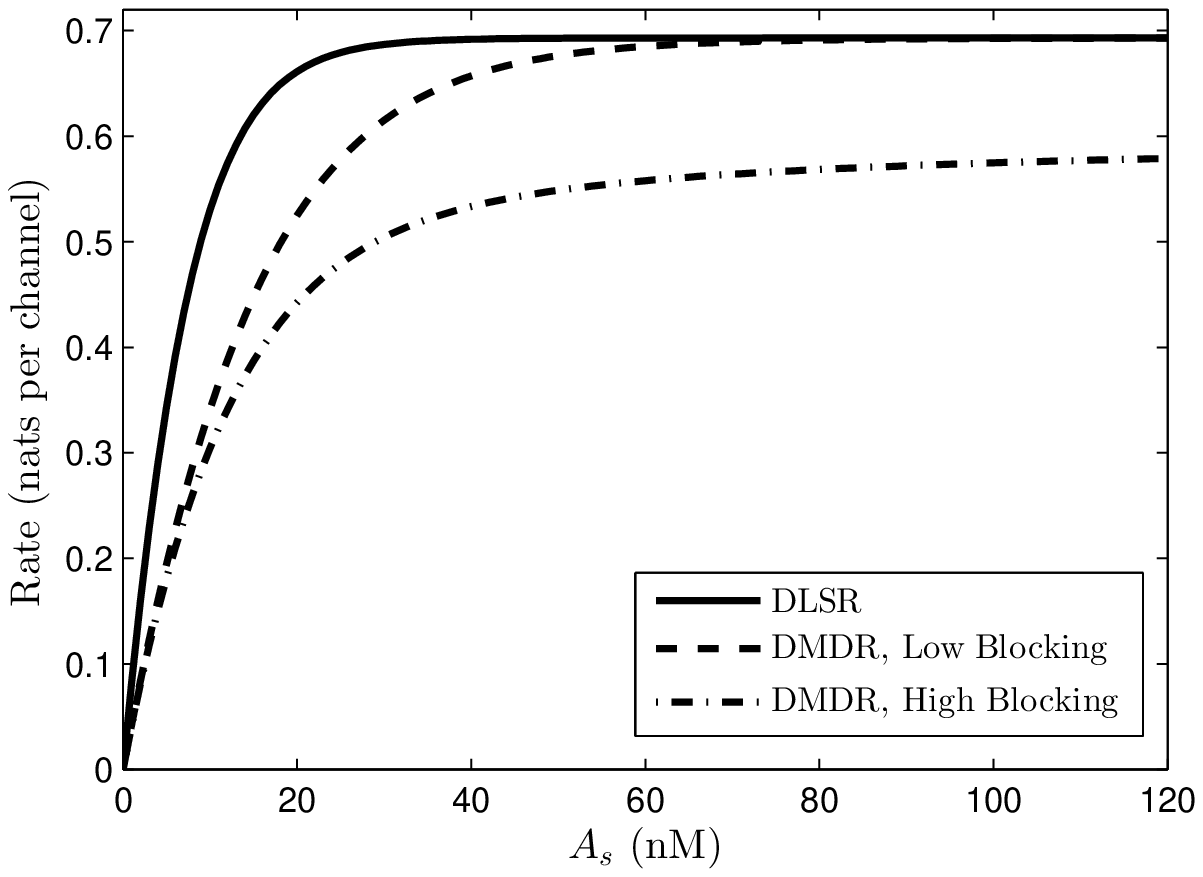}
\caption{Maximum achievable equal rates by viewing interference as noise for DLSR and DMDR with binary inputs for $\alpha_{s_1}=\alpha_{s_2}=\frac{1}{2}$ and $A_{ne}=0$..}
\label{fig42-gs}
\vspace{-1em}
\end{minipage}
\end{figure}

\section{Coclusion}\label{conclusion}
In this paper, we first investigated the capacity performance of point-to-point communication scenarios, including Level and Type scenarios. We also modeled the blocking as a Markov process and derived the probabilities of binding and blocking. Next, we derived a new upper bound on the capacity of the BIC at low SNR-regime based on the KL divergence upper bound as well as a lower bound. As expected and confirmed by simulations, the blocking would decrease the capacity of type scenario. Then we proposed three scenarios for the multiple-access communication, including same types of molecules with Different Labeling and Same types of Receptors (DLSR), Same types of Molecules and Receptors (SMSR), and Different types of Molecules and Receptors (DMDR) scenarios and investigated their capacity region and total capacity. We derived some inner bounds on the capacity region of these scenarios when the environment noise is negligible. Based on numerical results, DLSR outperforms the other scenarios for all values of the maximum signal level from the total capacity point of view. For small values of the maximum signal level, SMSR has better performance than DMDR, whereas for large values of maximum signal level, DMDR has better performance.

\section{Acknowledge}\label{Acknowledge}
The authors would like to thank Dr. Amin Gohari, for his helpful comments.

\bibliographystyle{ieeetr}
\bibliography{reftest}
\appendices

\section{Proof of Theorem \ref{theorem1}}
\label{AppendixProoftheorem1}
We find KL upper bound for the BIC as follows:
\small
\begin{align*}
\nonumber
I(X;Y) &\leq \sum_{x,y} \left[P(x,y)-P(x)P(y)\right] \log{P(y|x)}\\
&=\sum_{x,y} \left[P(x,y) - P(x)P(y)\right] \log{\left({{N^\prime \choose y} f_{p_b}^y(x+A_{ne}) (1-f_{p_b}(x+A_{ne}))^{N^\prime-y}}\right)}\\
&=E_{P(x,y)}\left[\log{{N^\prime \choose y}}\right]-E_{P(x)P(y)}\left[\log{{N^\prime \choose y}}\right]+\E_{P(x,y)}\left[y\log{f_{p_b}(x+A_{ne})}\right]\\
&\quad -\E_{P(x)P(y)}\left[y\log{f_{p_b}(x+A_{ne})}\right]+\E_{P(x,y)}[(N^\prime-y)\log{(1-f_{p_b}(x+A_{ne}))}]\\
&\quad -\E_{P(x)P(y)}\left[(N^\prime-y)\log{(1-f_{p_b}(x+A_{ne}))}\right]\\
&=\E_{P(x,y)}\left[y\log{f_{p_b}(x+A_{ne})}\right]-\E_{P(x)P(y)}\left[y\log{f_{p_b}(x+A_{ne})}\right] \left[\E_{P(x,y)}\left[y\log{(1-f_{p_b}(x+A_{ne}))}\right]\right.\\
&\quad \left.-\E_{P(x)P(y)}\left[y\log{(1-f_{p_b}(x+A_{ne}))}\right]\right]\\
&=\sum_{x}\left(\left(\sum_y y P(y|x)\right)\log{\frac{f_{p_b}(x+A_{ne})}{1-f_{p_b}(x+A_{ne})}}P(x)\right)\\
&\quad-\left(\sum_{x}\left(\sum_y y P(y|x)\right)P(x)\right)\left(\sum_x \log{\frac{f_{p_b}(x+A_{ne})}{1-f_{p_b}(x+A_{ne})}}P(x)\right)\\
&=\E\left[N^\prime f_{p_b}(x+A_{ne})\log{\left(\frac{f_{p_b}(x+A_{ne})}{1-f_{p_b}(x+A_{ne})}\right)}\right]-\E \left[N^\prime f_{p_b}(x+A_{ne})\right]\E \left[\log{\left(\frac{f_{p_b}(x+A_{ne})}{1-f_{p_b}(x+A_{ne})}\right)}\right]\\
&=N^\prime \mathsf{Cov} \left(f_{p_b}(X+A_{ne}), \log{\left(\frac{f_{p_b}(X+A_{ne})}{(1-f_{p_b}(X+A_{ne}))}\right)}\right).
\end{align*}
\normalsize
As mentioned earlier, $f_{p_b}(X+A_{ne})$ is an increasing function. Hence,
$$\mathsf{Cov}{\left(f_{p_b}(X+A_{ne}), \log{ \left(\frac{f_{p_b}(X+A_{ne})}{(1-f_{p_b}(X+A_{ne}))}\right)}\right)}\geq 0$$.
A further observation is that
$$ C \leq \max_{P(x)} N^\prime \mathsf{Cov}{\left(f_{p_b}(X+A_{ne}), \log{\left(\frac{f_{p_b}(X+A_{ne})}{(1-f_{p_b}(X+A_{ne}))}\right)}\right)}$$ is always achievable with a binary random variable $X$. We consider two points, $x_1$ and $x_2$ with probabilities $p_1$ and $p_2$. We have
\begin{align*}
\max_{P(x)} \mathsf{Cov}{(f_{p_b}(X+A_{ne}), \log (F) )}=\max_{\substack{P(x):\\ \E(f_{p_b}(X+A_{ne}))\leq \alpha_s A_s^\prime,\\ 0 \leq X \leq A_s^\prime}}&(\E [f_{p_b}(X+A_{ne})\log{F}]-\E[f_{p_b}(X+A_{ne})]\E[\log{F}])\\
=\max_{\substack{P(x):\\ \E(f_{p_b}(X+A_{ne}))\leq \alpha_s A_s^\prime,\\ 0 \leq X \leq A_s^\prime}}&(\E[(f_{p_b}(X+A_{ne})-\E[f_{p_b}(X+A_{ne})])\log{F}]),
\end{align*}
where $F=\frac{f_{p_b}(X+A_{ne})}{1-f_{p_b}(X+A_{ne})}$. Now, based on the analysis in \cite[Appendix C]{aminian}, the optimal distribution is given by $P(x)=\frac{\alpha_s A_s^\prime}{f_{p_b}(A_s^\prime+A_{ne})}\delta(x-A_s^\prime)+\left(1-\frac{\alpha_s A_s^\prime}{f_{p_b}(A_s^\prime+A_{ne})}\right)\delta(x)$ and the upper bound is obtained as
\begin{align*}
\max_{\alpha_s A_s^\prime \leq f_{p_b}(\alpha_s A_s^\prime+A_{ne})}\frac{\alpha_s A_s^\prime}{f_{p_b}(A_s^\prime+A_{ne})}[f_{p_b}(A_s^\prime+A_{ne})-\alpha_s A_s^\prime]E,
\end{align*}
where $E=\log{\left(\frac{f_{p_b}(A_s^\prime+A_{ne})(1-f_{p_b}(A_{ne}))}{f_{p_b}(A_{ne})(1-f_{p_b}(A_s^\prime+A_{ne}))}\right)}$. The upper bound is equal to $$\frac{f_{p_b}(\alpha_s A_s^\prime+A_{ne})}{f_{p_b}(A_s^\prime+A_{ne})}[f_{p_b}(A_s^\prime+A_{ne})-f_{p_b}(\alpha_s A_s^\prime+A_{ne})]\log{\frac{f_{p_b}(A_s^\prime+A_{ne})(1-f_{p_b}(A_{ne}))}{f_{p_b}(A_{ne})(1-f_{p_b}(A_s^\prime+A_{ne}))}},$$ for $\alpha_s A_s^\prime \leq \frac{f_{p_b}(A_s^\prime+A_{ne})}{2}$ and $\frac{f_{p_b}(A_s^\prime+A_{ne})}{4}\log{\left(\frac{f_{p_b}(A_s^\prime+A_{ne})(1-f_{p_b}(A_{ne}))}{f_{p_b}(A_{ne})(1-f_{p_b}(A_s^\prime+A_{ne}))}\right)}$, otherwise.\\
Now, if we consider $f_{p_b}(X+A_{ne})=\frac{X+A_{ne}}{X+A_{ne}+\frac{\kappa}{\gamma}}$, then the upper bound is:
\begin{align*}
A_{\mathsf{Binomial}}(P(y|x))&:=\max_{\substack{P(x),\\ \E[X]=\alpha_s A_s^\prime,~ 0 \leq X \leq A_s^\prime}}\mathcal{U}(P(x,y))\\
&=N^\prime\begin{cases}\frac{f_{p_b}(\alpha_s A_s^\prime+A_{ne})}{f_{p_b}(A_s^\prime+A_{ne})}[f_{p_b}(A_s^\prime+A_{ne})-f_{p_b}(\alpha_s A_s^\prime+A_{ne})]E,&\textit{if}\quad(*),
\\
\frac{f_{p_b}(A_s^\prime+A_{ne})}{4}E,&\textit{if}\quad (**),
\end{cases}
\end{align*}
where $E=\log{\left(\frac{f_{p_b}(A_s^\prime+A_{ne})(1-f_{p_b}(A_{ne}))}{f_{p_b}(A_{ne})(1-f_{p_b}(A_s^\prime+A_{ne}))}\right)}$, $(*):f_{p_b}(\alpha_s A_s^\prime+A_{ne})< \frac{f_{p_b}(A_s^\prime+A_{ne})}{2}$, and $(**):f_{p_b}(\alpha_s A_s^\prime+A_{ne})\geq \frac{f_{p_b}(A_s^\prime+A_{ne})}{2}$ .

\section{Proof of lemma \ref{lemma1}}
\label{AppendixProoflemma1}
Let
\begin{align*}
p_c=(1-p_b)^{N^\prime}.
\end{align*}
The BIC transition probabilities by considering binary input is characterized as
\begin{align*}
&P(y=0|x=0)=1, \quad P(y=i|x=0)=0,\qquad i=1,...,N^\prime, \\
&P(y=i|x=A_s^\prime)={N^\prime \choose i}p_b^i{\left(1-p_b\right)}^{N^\prime-i},\qquad i=1,...,N^\prime.
\end{align*}
Assume $P(x=A_s^\prime)=\alpha$. The average constraint results in $\alpha \leq \alpha_s$. The lower bound on the BIC capacity without considering the average constraint could be derived as follows:
\begin{align*}
C&=\max_{\alpha} I(X;Y)=\max_{\alpha} H(Y) - H(Y|X)\\
&=\max_{\alpha} H(Y) - P(x=0)H(Y|x=0)-P(x=A_s^\prime)H(Y|x=A_s^\prime)\\
&=\max_{\alpha} -\sum_{i =1}^{N^\prime}\alpha P(y=i|x=A_s^\prime)\log{(\alpha P(y=i|x=A_s^\prime))}\\
& \quad-(1-\alpha+\alpha p_c)\log{(1-\alpha+\alpha p_c)} - \alpha H(Y|x=A_s^\prime)\\
&=\max_{\alpha} -\alpha(1-p_c) \log{\alpha}+\alpha p_c \log{p_c} -(1-\alpha+\alpha p_c)\log{(1-\alpha+\alpha p_c)}.
\end{align*}
Taking a derivative with respect to $\alpha$ from the above expression and setting it to zero we obtain $\alpha^*=\frac{1}{1-p_c+e^{\frac{-p_c \log p_c}{1-p_c}}}$.
Then,
\begin{align*}
C=H{\left(\frac{1}{1+e^{g(p_c)}}\right)}-\frac{g(p_c)}{1+e^{g(p_c)}},
\end{align*}
where $H(p)=-p\log p-(1-p) \log{(1-p)}$ and $g(p)=\frac{H(p)}{1-p}$. Now, if we consider the average constraint, the equation for $C$ is valid for $\alpha^* \leq \alpha_s$ since the mutual information is concave in $\alpha$. But for $\alpha^*>\alpha_s$, the capacity lower bound is obtained for $\alpha=\alpha_s$.

\section{Proof of Lemma \ref{lemma2}}
\label{AppendixProoflemma2}
We prove the lemma for the DLSR scenario. The approach for the DMDR scenario is the same. Let
\begin{align*}
&p_{b_{11}}=p_{b_1}^{DLSR}(x_1=A_{s_1},X_2=x_2), \quad p_{b_{21}}=p_{b_2}^{DLSR}(x_1=A_{s_1},X_2=x_2),\\
&p_{c_1}=P(y_1=0|x_1=A_{s_1},X_2=x_2)=\left(\frac {x_2+\frac{\kappa}{\gamma}}{A_{s_1}+x_2+\frac{\kappa}{\gamma}}\right)^{nN},\\
&p_{c_{10}}=P(y_1=0|x_1=A_{s_1},x_2=0)=\left(\frac{\frac{\kappa}{\gamma}}{A_{s_1}+\frac{\kappa}{\gamma}}\right)^{nN}.
\end{align*}
Channel transition probabilities for the first transmitter by considering binary inputs $x_1 \in \{ 0, A_{s_1} \}$ and $x_2 \in \{ 0, A_{s_2} \}$ are characterized as
\begin{align*}
&P(y_1=0|x_1=0, X_2=x_2)=1,\\
&P(y_1=i|x_1=A_{s_1}, X_2=x_2)=\sum_{j=0}^{nN-i} \binom{nN}{i} \binom{nN-i}{j} p_{b_{11}}^i p_{b_{21}}^j (1-(p_{b_{11}}+p_{b_{21}}))^{nN-i-j}, \qquad i=0,...,nN,
\end{align*}
Assume $P(x_1=A_{s_1})=\alpha_1$. The average constraint for the first transmitter results in $\alpha_1 \leq \alpha_{s_1}$. In the following, the maximum achievable individual rate for the first transmitter, $C_1$, is computed. The approach for computing $C_2$ is the same. Without considering the average constraint we have
\begin{align*}
C_1 &=\max_{\substack{x_2,~\alpha_1}} I(X_1;Y_1|X_2=x_2)=\max_{\substack{x_2,~\alpha_1}} H(Y_1|X_2=x_2)-H(Y_1|X_1,X_2=x_2)\\
&=\max_{\substack{x_2,~\alpha_1}} H(Y_1|X_2=x_2)-P(x_1=0) H(Y_1|x_1=0,X_2=x_2)-P(x_1=A_{s_1}) H(Y_1|x_1=A_{s_1},X_2=x_2)\\
&=\max_{\substack{x_2,~\alpha_1}} -\sum_{i=1}^{nN} \alpha_1 P(y_1=i|x_1=A_{s_1},X_2=x_2) \log{(\alpha_1 P(y_1=i|x_1=A_{s_1},X_2=x_2))}\\
&\quad -(1-\alpha_1+\alpha_1 p_{c_1})\log{(1-\alpha_1+\alpha_1 p_{c_1})}-\alpha_1 H(Y_1|x_1=A_{s_1},X_2=x_2)\\
&=\max_{\substack{x_2,~\alpha_1}} -\alpha_1(1-p_{c_1})\log{\alpha_1}+\alpha_1 p_{c_1} \log{p_{c_1}}-(1-\alpha_1+\alpha_1 p_{c_1})\log{(1-\alpha_1+\alpha_1 p_{c_1})}.
\end{align*}
Taking a derivative with respect to $\alpha_1$ from the above expression and setting it to zero we obtain $\alpha_1^*=\frac{1}{1-p_{c_{1}}+e^{\frac{-p_{c_{1}} \log p_{c_{1}}}{1-p_{c_{1}}}}}$. Then,
\begin{align*}
&C_1=\max_{\substack{x_2}} \left(H\left(\frac{1}{1+e^{g({p_c}_1)}}\right)-\frac{g(p_{c_1})}{1+e^{g(p_{c_1})}}\right),
\end{align*}
where $H(p)=-p\log p-(1-p) \log{(1-p)}, \ g(p)=\frac{H(p)} {1-p}$. Taking a derivative with respect to $x_2$ from the above expression we obtain
\begin{align*}
\frac {d}{d x_2}I_{\alpha_1^*}(X_1;Y|X_2=x_2)=-\frac{p_{c_1}^\prime g^\prime(p_{c_1}) e^{g(p_{c_1})}}{(1+e^{g(p_{c_1})})^2}H^\prime \left(\frac{1}{1+e^{g(p_{c_1})}}\right)-\frac{p_{c_1}^\prime g^\prime(p_{c_1}) (1+e^{g(p_{c_1})}-g(p_{c_1})e^{g(p_{c_1})})}{(1+e^{g(p_{c_1})})^2}.
\end{align*}
Since $H^\prime (p)=\log(\frac{1-p}{p})$, we have
\begin{align*}
\frac {d}{d x_2}I_{\alpha_1^*}(X_1;Y|X_2=x_2)=-\frac{p_{c_1}^\prime g^\prime(p_{c_1}) (g(p_{c_1})e^{g(p_{c_1})}+1+e^{g(p_{c_1})}-g(p_{c_1})e^{g(p_{c_1})})}{(1+e^{g(p_{c_1})})^2}=-\frac{p_{c_1}^\prime g^\prime(p_{c_1}) (1+e^{g(p_{c_1})})}{(1+e^{g(p_{c_1})})^2},
\end{align*}
and this is a negetive value for all $x_2 \geq 0$ since $p_{c_1}^\prime=\frac{nN A_{s_1}}{A_{s_1}+x_2+\frac{\kappa}{\gamma}}{\left(\frac{x_2+\frac{\kappa}{\gamma}}{{A_{s_1}+x_2+\frac{\kappa}{\gamma}}}\right) }^{nN-1}>0$ and  $g^\prime(p_{c_1})=-\frac{p_{c_1}^\prime \log{p_{c_1}}}{(1+e^{g(p_{c_1})})^2}>0$. $x_2$ can take two values $0$ and $A_{s_2}$. So the maximum occurs when $x_2=0$. Hence,
\begin{align*}
&C_1=H{\left(\frac{1}{1+e^{g(p_{c_{10}})}}\right)}-\frac{g(p_{c_{10}})}{1+e^{g(p_{c_{10}})}}.
\end{align*}
Now, we consider the average constraint. For both values of $x_2=0$ and  $x_2=A_{s_2}$, if $\alpha_{s_1} \geq \alpha_1^*(x_2)$, the maximum for $I(X_1;Y_1|x_2=x_2)$ occurs when $\alpha_1=\alpha_1^*(x_2)$ and if $0<\alpha_{s_1}<\alpha_1^*(x_2)$, the maximum occurs when $\alpha_1=\alpha_{s_1}$ since $I(X_1,Y_1|X_2=x_2)$ is concave in $\alpha_1$. Let $\alpha_{10}^*=\alpha_1^*(x_2=0)$ and $\alpha_{11}^*=\alpha_1^*(x_2=A_{s_2})$.
If $\alpha_{s_1} \geq \alpha_{10}^*$ and $\alpha_{s_1} \geq \alpha_{11}^*$, $C_1=\max \{ I_{\alpha_{10}^*}(X_1;Y_1|x_2=0), I_{\alpha_{11}^*}(X_1;Y_1|x_2=A{s_2}) \}$ equals to $I_{\alpha_{10}^*}(X_1;Y_1|x_2=0)$. 
If $\alpha_{s_1} \geq \alpha_{10}^*$ and $\alpha_{s_1}<\alpha_{11}^*$, $C_1=\max \{ I_{\alpha_{10}^*}(X_1;Y_1|x_2=0), I_{\alpha_{s_{1}}}(X_1;Y_1|x_2=A{s_2}) \}$ equals to $I_{\alpha_{10}^*}(X_1;Y_1|x_2=0)$ since $I_{\alpha_{10}^*}(X_1;Y_1|x_2=0) > I_{\alpha_{11}^*}(X_1;Y_1|x_2=A_{s_2}) \geq I_{\alpha_{s_1}}(X_1;Y_1|x_2=A_{s_2})$. 
If $0<\alpha_{s_1}<\alpha_{10}^*$ and $0<\alpha_{s_1}<\alpha_{11}^*$, $C_1=\max \{ I_{\alpha_{s_1}}(X_1;Y_1|x_2=0), I_{\alpha_{s_{1}}}(X_1;Y_1|x_2=A{s_2}) \}$ equals to $I_{\alpha_{s_1}}(X_1;Y_1|x_2=0)$ since $$\frac {d}{d x_2}I(X_1;Y|X_2=x_2)=\alpha_1 p_{c_1}^\prime \log{\frac{\alpha_1 p_{c_1}}{1-\alpha_1+\alpha_1 p_{c_1}}} \leq 0$$ and $I(X_1;Y_1|X_2=x_2)$ is a decreasing function with respect to $x_2$ for all values of $\alpha_1 \in [0,1]$. 
If $0<\alpha_{s_1}<\alpha_{10}^*$ and $\alpha_{s_1} \geq \alpha_{11}^*$, $C_1=\max \{ I_{\alpha_{s_1}}(X_1;Y_1|x_2=0), I_{\alpha_{11}^*}(X_1;Y_1|x_2=A{s_2}) \}$ equals to $I_{\alpha_{s_1}}(X_1;Y_1|x_2=0)$ since $I_{\alpha_{s_1}}(X_1;Y_1|x_2=0)>I_{\alpha_{11}^*}(X_1;Y_1|x_2=0)>I_{\alpha_{11}^*}(X_1;Y_1|x_2=A_{s_2})$.

\section{Proof of Lemma \ref{lemma3}}
\label{AppendixProoflemma3}
We prove the lemma for the DLSR scenario. The approach for the DMDR scenario is the same. Let
\begin{align*}
&p_{b_{11}}=p_{b_1}^{DLSR}(x_1=A_{s_1},X_2=x_2), \quad p_{b_{21}}=p_{b_2}^{DLSR}(x_1=A_{s_1},X_2=x_2),\\
&p_{b_{12}}=p_{b_1}^{DLSR}(X_1=x_1,x_2=A_{s_2}), \quad p_{b_{22}}=p_{b_2}^{DLSR}(X_1=x_1,x_2=A_{s_2}),\\
&p_{c_{10}}=P(y_i=0|x_1=A_{s_1},x_2=0)=\left(\frac{\frac{\kappa}{\gamma}}{A_{s_1}+\frac{\kappa}{\gamma}}\right)^{nN},\\
&p_{c_{11}}=P(y_i=0|x_1=A_{s_1},x_2=A_{s_2})=\left(\frac{A_{s_2}+\frac{\kappa}{\gamma}}{A_{s_i}+A_{s_j}+\frac{\kappa}{\gamma}}\right)^{nN},
\end{align*}
Channel transition probabilities by considering binary inputs $x_1 \in \{ 0, A_{s_1} \}$ and $x_2 \in \{ 0, A_{s_2} \}$ are characterized as
\begin{align*}
&P(y_1=0|x_1=0, X_2=x_2)=P(y_2=0|X_1=x_1, x_2=0)=1,\\
&P(y_1=i|x_1=A_{s_1}, X_2=x_2)=\sum_{j=0}^{nN-i} \binom{nN}{i} \binom{nN-i}{j} p_{b_{11}}^i p_{b_{21}}^j (1-(p_{b_{11}}+p_{b_{21}}))^{nN-i-j}, \qquad i=0,...,nN,\\
&P(y_2=i|X_1=x_1, x_2=A_{s_2})=\sum_{j=0}^{nN-i} \binom{nN}{i} \binom{nN-i}{j} p_{b_{12}}^j p_{b_{22}}^i (1-(p_{b_{12}}+p_{b_{22}}))^{nN-i-j}, \qquad i=0,...,nN.
\end{align*}
Assume $P(x_1=A_{s_1})=\alpha_1$ and $P(x_2=A_{s_2})=\alpha_2$. The average constraints result in $\alpha_1 \leq \alpha_{s_1}$ and $\alpha_2 \leq \alpha_{s_2}$. The interference-as-noise inner bound for this channel is computed as follows:
\begin{align*}
R_1 &< I(X_1;Y_1)= H(Y_1)-H(Y_1|X_1)=H(Y_1)-P(x_1=0)H(Y_1|x_1=0)-P(x_1=A_{s_1})H(Y_1|x_1=A_{s_1})\\
&= -\sum_{i=1}^{nN} \alpha_1((1-\alpha_2) P(y_1=i|x_1=A_{s_1},x_2=0)+\alpha_2 P(y_1=i|x_1=A_{s_1},x_2=A_{s_1}))\\
&\quad \times\log{(\alpha_1((1-\alpha_2) P(y_1=i|x_1=A_{s_1},x_2=0)+\alpha_2 P(y_1=i|x_1=A_{s_1},x_2=A_{s_2})))}\\
&\quad -((1-\alpha_2) ((1-\alpha_1) +\alpha_1 p_{c_{10}})+\alpha_2((1-\alpha_1) +\alpha_1 p_{c_{11}}))\\
&\quad \times \log{((1-\alpha_2) ((1-\alpha_1) +\alpha_1 p_{c_{10}})+\alpha_2((1-\alpha_1) +\alpha_1 p_{c_{11}}))}-\alpha_1 H(Y_1|x_1=A_{s_1})
\end{align*}
\begin{align*}
&=-\alpha_1 \log{\alpha_1}\sum_{i=1}^{nN} ((1-\alpha_2) P(y_1=i|x_1=A_{s_1},x_2=0)+\alpha_2 P(y_1=i|x_1=A_{s_1},x_2=A_{s_2}))\\
&\quad +\alpha_1 (H(Y_1|x_1=A_{s_1})+((1-\alpha_2) p_{c_{10}}+\alpha_2 p_{c_{11}}) \log{((1-\alpha_2) p_{c_{10}}+\alpha_2 p_{c_{11}})})\\
&\quad -((1-\alpha_1)+\alpha_1((1-\alpha_2) p_{c_{10}}+\alpha_2 p_{c_{11}}))\log{((1-\alpha_1)+\alpha_1((1-\alpha_2) p_{c_{10}}+\alpha_2 p_{c_{11}}))}\\
&\quad -\alpha_1 H(Y_1|x_1=A_{s_1})\\
&=-\alpha_1(1-(1-\alpha_2) p_{c_{10}}-\alpha_2 p_{c_{11}})\log{\alpha_1}+\alpha_1((1-\alpha_2) p_{c_{10}}+\alpha_2 p_{c_{11}}) \log{((1-\alpha_2) p_{c_{10}}+\alpha_2 p_{c_{11}})}\\
&\quad -\alpha_1 \left(\frac{1-\alpha_1}{\alpha_1}+(1-\alpha_2) p_{c_{10}}+\alpha_2 p_{c_{11}}\right)\log{\left(\alpha_1 \left(\frac{1-\alpha_1}{\alpha_1}+(1-\alpha_2) p_{c_{10}}+\alpha_2 p_{c_{11}}\right)\right)}\\
&=-\log{\alpha_1}+\alpha_1((1-\alpha_2) p_{c_{10}}+\alpha_2 p_{c_{11}})\log{((1-\alpha_2) p_{c_{10}}+\alpha_2 p_{c_{11}})}\\
&\quad -\alpha_1 \left(\frac{1-\alpha_1}{\alpha_1}+(1-\alpha_2) p_{c_{10}}+\alpha_2 p_{c_{11}}\right)\log{\left(\frac{1-\alpha_1}{\alpha_1}+(1-\alpha_2) p_{c_{10}}+\alpha_2 p_{c_{11}}\right)}.
\end{align*}
for some $\alpha_1 \in [0,\alpha_{s_1}]$ and $\alpha_2 \in [0, \alpha_{s_2}]$. With the same approach for $R_2$ we have
\begin{align*}
R_2 &< -\log{\alpha_2}+\alpha_2((1-\alpha_1) p_{c_{20}}+\alpha_1 p_{c_{21}})\log{((1-\alpha_1) p_{c_{20}}+\alpha_1 p_{c_{21}})}\\
&\quad -\alpha_2 \left(\frac{1-\alpha_2}{\alpha_2}+(1-\alpha_1) p_{c_{20}}+\alpha_1 p_{c_{21}}\right)\log{\left(\frac{1-\alpha_2}{\alpha_2}+(1-\alpha_1) p_{c_{20}}+\alpha_1 p_{c_{21}}\right)},
\end{align*}
For $A_{s_1}=A_{s_2}=A_s$, we have $p_{c_{10}}=p_{c_{20}}$ and $p_{c_{11}}=p_{c_{21}}$. The points where $R_1=R_2$, without considering the average constraints, are as follows:
\begin{align*}
R_1=R_2&=k \max_{\substack{\alpha}}-\log{\alpha}+\alpha((1-\alpha) p_{c_{10}}+\alpha p_{c_{11}}) \log{((1-\alpha) p_{c_{10}}+\alpha p_{c_{11}})}\\
&\quad -\alpha\left(\frac{1-\alpha}{\alpha}+(1-\alpha) p_{c_{10}}+\alpha p_{c_{11}}\right)\log{\left(\frac{1-\alpha}{\alpha}+(1-\alpha) p_{c_{10}}+\alpha p_{c_{11}}\right)},
\end{align*}
for some $k \in [0,1]$. Taking a derivative with respect to $\alpha$ from the above expression and setting it to zero we obtain
\begin{align*}
&((1-2 \alpha)p_{c_{10}}+2\alpha p_{c_{11}})\log{((1-\alpha) p_{c_{10}}+\alpha p_{c_{11}})}\\
&\quad-((1-2\alpha)p_{c_{10}}+2\alpha p_{c_{11}}-1)\log{\left(\frac{1-\alpha}{\alpha}+(1-\alpha) p_{c_{10}}+\alpha p_{c_{11}}\right)}=0.
\end{align*}
If we consider the average constraints with $\alpha_{s_1}=\alpha_{s_2}=\alpha_s$, the above equation for the optimum value of $\alpha$ is valid if the solution of the equation is lower than or equal to $\alpha_s$ since $I(X_1;Y_1)$ for $\alpha_1=\alpha_2=\alpha$ is concave in $\alpha$. If the solution is higher than $\alpha_s$, the maximum occurs when $\alpha=\alpha_s$.
\section{Proof of Lemma \ref{lemma8}}
\label{AppendixProoflemma8}
Let
\begin{align*}
&{p_b}_{10}=p_b^{SMSR}(x_1=A_{s_1},x_2=0), \quad {p_b}_{01}=p_b^{SMSR}(x_1=0,x_2=A_{s_2}),
\end{align*}
\begin{align*}
&p_{b_{11}}=p_b^{SMSR}(x_1=A_{s_1},x_2=A_{s_2}),\\
&p_{c_{10}}=P(y=0|x_1=A_{s_1},x_2=0)=\left(\frac{\frac{\kappa}{\gamma}}{A_{s_1}+\frac{\kappa}{\gamma}}\right)^{nN},
\end{align*}
Channel transition probabilities by considering binary inputs $x_1 \in \{ 0, A_{s_1} \}$ and $x_2 \in \{ 0, A_{s_2} \}$ are characterized as
\begin{align*}
&P(y=0|x_1=0,x_2=0)=1,\\
&P(y=i|x_1=A_{s_1}, x_2=0)=\binom{nN}{i} {{p_b}_{10}}^i {(1-{p_b}_{10})}^{nN-i}, \qquad i=0,...,nN,\\
&P(y=i|x_1=0, x_2=A_{s_2})=\binom{nN}{i} {{p_b}_{01}}^i {(1-{p_b}_{01})}^{nN-i}, \qquad i=0,...,nN,\\
&P(y=i|x_1=A_{s_1}, x_2=A_{s_2})=\binom{nN}{i} {p_{b_{11}}}^i {(1-p_{b_{11}})}^{nN-i}, \qquad i=0,...,nN.
\end{align*}
Assume $P(x_1=A_{s_1})=\alpha_1$. The average constraint for the first transmitter results in $\alpha_1 \leq \alpha_{s_1}$. In the following, the maximum achievable individual rate for the first transmitter, $C_1$, is computed. The approach for computing $C_2$ is the same. Without considering the average constraint we have
\begin{align*}
C_1 &=\max_{\substack{x_2,~\alpha_1}} I(X_1;Y|X_2=x_2)=\max_{\substack{x_2,~\alpha_1}} H(Y|X_2=x_2)-H(Y|X_1,X_2=x_2)\\
&=\max_{\substack{x_2,~\alpha_1}} H(Y|X_2=x_2)-P(x_1=0) H(Y|x_1=0,X_2=x_2)-P(x_1=A_{s_1}) H(Y|x_1=A_{s_1},X_2=x_2)\\
&=\max_{\substack{x_2,~\alpha_1}} -\sum_{i=0}^{N} P(y=i|X_2=x_2) \log {P(y=i|X_2=x_2)} -(1-\alpha_1) H(Y|x_1=0,X_2=x_2)\\
&\quad -\alpha_1 H(Y|x_1=A_{s_1},X_2=x_2).
\end{align*}
We can write this as follows:
\begin{align*}
&C_1=\max \{c_{10},c_{11}\},\\ &c_{10}=\max_{\substack{\alpha_1}} I(X_1;Y|x_2=0),\quad c_{11}=\max_{\substack{\alpha_1}} I(X_1;Y|x_2=A_{s_2}).
\end{align*}
For $c_{10}$ we have
\begin{align*}
c_{10}&=\max_{\substack{\alpha_1}} -\sum_{i=0}^{nN} P(y=i|x_2=0) \log {P(y=i|x_2=0)}-\alpha_1 H(Y|x_1=A_{s_1},x_2=0)\\
&=\max_{\substack{\alpha_1}} -\sum_{i=1}^{nN} \alpha_1 P(y=i|x_1=A_{s_1},x_2=0) \log {(\alpha_1 P(y=i|x_1=A_{s_1},x_2=0))}\\
&\quad -(1-\alpha_1+\alpha_1 p_{c_{10}}) \log{(1-\alpha_1+\alpha_1 p_{c_{10}})}-\alpha_1 H(Y|x_1=A_{s_1},x_2=0)\\
&= \max_{\substack{\alpha_1}} -\alpha_1(1-p_{c_{10}}) \log{\alpha_1}+\alpha_1 p_{c_{10}} \log{p_{c_{10}}}-(1-\alpha_1+\alpha_1 p_{c_{10}}) \log{(1-\alpha_1+\alpha_1 p_{c_{10}})}.
\end{align*}
Taking a derivative with respect to $\alpha_1$ from the above expression and setting it to zero we obtain
$\alpha_{10}^*=\frac{1}{1-p_{c_{10}}+e^{-\frac{p_{c_{10}} \log{p_{c_{10}}}}{1-p_{c_{10}}}}}$,
Then,
\begin{align*}
c_{10}= H{\left(\frac{1}{1+e^{g(p_{c_{10}})}}\right)}-\frac{g(p_{c_{10}})}{1+e^{g(p_{c_{10}})}},
\end{align*}
where $H(p)=-p\log p-(1-p) \log{(1-p)}, \ g(p)=\frac{H(p)} {1-p}$. If we consider the average constraint, the above equation for $c_{10}$ is valid if $\alpha_{s_1} \geq \alpha_{10}^*$ since $I(X_1;Y_1|x_2=0)$ is concave in $\alpha_1$. If $0 <\alpha_{s_1}<\alpha_{10}^*$, the maximum occurs when $\alpha_1=\alpha_{s_1}$.
For $c_{11}$ we have
\begin{align*}
c_{11}&=\max_{\substack{\alpha_1}} I(X_1;Y|x_2=A_{s_2})=\max_{\substack{\alpha_1}} -\sum_{i=0}^{nN} P(y=i|x_2=A_{s_2}) \log {P(y=i|x_2=A_{s_2})} \\
&\quad -(1-\alpha_1) H(Y|x_1=0,x_2=A_{s_2})-\alpha_1 H(Y|x_1=A_{s_1},x_2=A_{s_2})\\
&=\max_{\substack{\alpha_1}}-\sum_{i=0}^{nN} ((1-\alpha_1) P(y=i|x_1=0,x_2=A_{s_2})+\alpha_1 P(y=i|x_1=A_{s_1},x_2=A_{s_2}))\\
&\quad \times \log{((1-\alpha_1) P(y=i|x_1=0,x_2=A_{s_2})+\alpha_1 P(y=i|x_1=A_{s_1},x_2=A_{s_2}))}\\
&\quad -(1-\alpha_1) H(Y|x_1=0,x_2=A_{s_2})-\alpha_1 H(Y|x_1=A_{s_1},x_2=A_{s_2})\\
&=\max_{\substack{\alpha_1}}-\sum_{i=0}^{nN} \left[(1-\alpha_1) P(y=i|x_1=0,x_2=A_{s_2})\log{\left((1-\alpha_1)+\alpha_1 \frac{P(y=i|x_1=A_{s_1},x_2=A_{s_2})}{P(y=i|x_1=0,x_2=A_{s_2})}\right)}\right.\\
&\quad \left.+\alpha_1 P(y=i|x_1=A_{s_1},x_2=A_{s_2}) \log{\left((1-\alpha_1) \frac{P(y=i|x_1=0,x_2=A_{s_2})}{P(y=i|x_1=A_{s_1},x_2=A_{s_2})}+\alpha_1 \right)}\right].
\end{align*}
Taking a derivative with respect to $\alpha_1$ from the above expression and setting it to zero we obtain
\begin{align*}
\sum_{i=0}^{nN} &\left[P(y=i|x_1=0,x_2=A_{s_2})\log{\left((1-\alpha_1)+\alpha_1 \frac{P(y=i|x_1=A_{s_1},x_2=A_{s_2})}{P(y=i|x_1=0,x_2=A_{s_2})}\right)}\right.\\
&\quad \left.-P(y=i|x_1=A_{s_1},x_2=A_{s_2}) \log{\left((1-\alpha_1) \frac{P(y=i|x_1=0,x_2=A_{s_2})}{P(y=i|x_1=A_{s_1},x_2=A_{s_2})}+\alpha_1 \right)}\right]=0.
\end{align*}
If we consider the average constraint, the above equation for the optimum value of $\alpha_1$ is valid when the solution of the equation is lower than or equal to $\alpha_{s_1}$ since $I(X_1;Y_1|x_2=A_{s_2})$ is concave in $\alpha_1$. If the solution is higher than $\alpha_{s_1}$, the maximum occurs when $\alpha_1=\alpha_{s_1}$.
\end{document}